\documentclass[a4paper]{article}

\usepackage[english]{babel}
\usepackage[square,sort,comma,numbers]{natbib}
\usepackage[utf8]{inputenc}
\usepackage{amsmath}
\usepackage{amsfonts}
\usepackage{amsthm}
\usepackage{graphicx}
\usepackage{subcaption}
\usepackage{color}
\usepackage{textcomp}
\usepackage{hyperref}
\hypersetup{
    colorlinks=true,
    linkcolor=blue,
    filecolor=magenta,      
    urlcolor=cyan,
}
\usepackage{algorithmicx}
\usepackage{algpseudocode}
\usepackage{algorithm}
\usepackage{tcolorbox}

\begin{document}
\newtheorem{defn}{Definition}
\newtheorem{lem}{Lemma}
\newtheorem{cor}{Corollary}
\newtheorem{thm}{Theorem}
\newtheorem{cex}{Counterexample}
\newtheorem{rem}{Remark}
\algnewcommand\algorithmicinput{\textbf{INPUT:}}
\algnewcommand\INPUT{\item[\algorithmicinput]}
\algnewcommand\algorithmicoutput{\textbf{OUTPUT:}}
\algnewcommand\OUTPUT{\item[\algorithmicoutput]}
\algnewcommand\algorithmicbreak{\textbf{break}}
\algnewcommand\Break{\item[\algorithmicbreak]}
\newcommand*{\LargerCdot}{\raisebox{-0.25ex}{\scalebox{1.5}{$\cdot$}}}
\title{The K Shortest Paths Problem with Application to Routing}
\author{David Burstein\textsuperscript{1,2} \and Leigh Metcalf \textsuperscript{1}}
\date{\today}
\baselineskip 1.5em
\maketitle
\renewcommand{\thefootnote}{\fnsymbol{footnote}}
\footnotetext{1. Software Engineering Institute.  Carnegie Mellon University. 4500 Fifth Avenue.\newline Pittsburgh, PA 15213}
\footnotetext{2. Department of Mathematics and Statistics.  Swarthmore College.  500 College Avenue,   Swarthmore, PA 19081}
\footnotetext{3. Questions or comments may be sent to dburste1@swarthmore.edu}
\footnotetext{4. [Distribution Statement A] This material has been approved for public release and unlimited distribution.  Please see Copyright notice for non-US government use and distribution.}
\bibliographystyle{plain}
\pagestyle{myheadings}
\thispagestyle{plain}
\markboth{D. Burstein and L. Metcalf}{The K Shortest Paths Problem}
\begin{abstract}
Due to the computational complexity of finding almost shortest simple paths, we propose that identifying a larger collection of (nonbacktracking) paths is more efficient than finding almost shortest simple paths on positively weighted real-world networks.  First, we present an easy to implement $O(m\log m+kL)$ solution for finding all (nonbacktracking) paths with bounded length $D$ between two arbitrary nodes on a positively weighted graph, where $L$ is an upperbound for the number of nodes in any of the $k$ outputted paths.  Subsequently, we illustrate that for undirected Chung-Lu random graphs, the ratio between the number of nonbacktracking and simple paths asymptotically approaches $1$ with high probability for a wide range of parameters.  We then consider an application to the almost shortest paths algorithm to measure path diversity for internet routing in a snapshot of the Autonomous System graph subject to an edge deletion process.  
\end{abstract}

\textbf{Keywords:}k shortest paths, internet routing, path sampling, edge deletion, simple paths, random graphs
\newline\newline
\indent \textbf{MSC:} 05C38, 05C85, 68R10, 90C35

\section{Introduction}
\label{intro}
Calculating almost shortest simple paths between two nodes on positively weighted graphs arises in many applications; such applications include, inferring the spreading path of a pathogen in a social network \cite{onnela2012spreading}, proposing novel complex relationships between biological entities \cite{frijters2010literature,he2011mining,shih2012single}, identifying membership of hidden communities in a graph \cite{palla2005uncovering,porter2016dynamical} and routing in the Autonomous System (AS) graph, as discussed in this work and others \cite{leguay2005describing,savage1999end}.  Even though research has demonstrated through simulation that exponentially slow solutions to the almost shortest simple paths problem often out-perform their polynomial time counterparts \cite{hershberger2007finding}, very little work has explored how properties in these empirically observed networks suggest efficient solutions to finding the almost shortest simple paths.  Consequently, we address this gap by proposing an efficient method in finding almost shortest simple paths for a specific family of graphs that emulate many features found in real-world networks.

When considering solutions to the almost shortest paths problem on a real-world network, as opposed to an arbitrary graph, such a solution should exploit the small diameter and locally tree-like properties of the graph. Furthermore, the number of paths between two fixed nodes grows exponentially in terms of path length.  The former property emphasizes that the optimal complexity for finding explicit representations for the $k$ shortest simple paths between two nodes, should roughly be $O(m+kL)$, where the graph has $m$ edges and paths consist of at most $L$ nodes. In addition, the latter property highlights the fact that constructing a set of just the $k$ shortest paths could ultimately exclude many paths of equal length.  To efficiently find almost shortest simple paths for real world networks, we consider the problem of identifying nonbacktracking paths, which allows us to revisit nodes under certain constraints.  We then provide theoretical and numeric results identifying conditions in the Chung-Lu random graph model such that the number of simple paths of prescribed length approximates the number of nonbacktracking paths.  As a result, we can construct an efficient algorithm for finding almost simple shortest paths by identifying a slightly larger set of nonbacktracking paths and deleting the paths that are not simple.
 
We emphasize that while our choice for considering the Chung-Lu model may appear arbitrary, Chung-Lu random graphs are closely related to the Stochastic Kronecker Graph model \cite{pinarsimilarity2014,chakrabarti2004r,leskovec2007scalable,leskovec2010kronecker}, a commonly used random graph model for evaluating the efficiency of graph algorithms \cite{bader2008snap,edmonds2010space,vineet2009fast}. Additionally, Chung-Lu random graphs  emulate many of the properties observed in real world networks; more specifically, realizations possess a small diameter along with degree heterogeneity. We also anticipate that our results carry over for other random graph models that are also locally tree-like \cite{castellano2017topological}.   While the last statement may appear obvious, proving precise theoretical upperbounds for the ratio between the number of simple and non-simple paths, is intimately related to constructing asymptotics for the dominating eigenvalue of the adjacency matrix, a highly nontrivial problem \cite{chung2003spectra,restrepo2007approximating,van2010graph,burstein2017asymptotics}.\newline
\indent In application, many existing algorithms for identifying the $k$ shortest simple paths are not designed to exploit properties often found in real world networks. Such approaches often require deleting edges from the graph and running a shortest path algorithm on the newly formed graph. Recalling that $n$ is the number of nodes and $m$ is the number of edges in a graph, Yen \cite{yen1971finding} provides an $O(kn[m+n\log n])$ solution that works for weighted, directed graphs, while Katoh \cite{katoh1982efficient} and Roditty \cite{roditty2005replacement}  provide $O(k[m+n\log n])$ and $O(km\sqrt{n})$  solutions respectively for undirected graphs.  More recently, Bernstein \cite{bernstein2010nearly} provides an  $O(km/\epsilon)$ algorithm for computing approximate replacement paths.  But since we want to calculate {\em{many}} paths, such solutions can be asymptotically expensive.
\newline \indent In contrast, the asymptotic performance for computing the $k$ shortest paths is more appropriate for implementation on real world networks.  Eppstein provides both $O(m +n\log n +k\log k + Lk)$ and  $O(m +n\log n + Lk)$  solutions \cite{eppstein1994finding,eppstein1998finding} for finding explicit representations of the $k$ shortest paths between two nodes, where $L$ is an upperbound on the number of nodes that appear in a path.  Nevertheless, recent variations to Eppstein's solution often emphasize the asymptotically inferior version, due to the sophistication and large constant factor behind the $O(m +n\log n + Lk)$ solution \cite{frederickson1993optimal,jimenez1999computing,jimenez2003lazy,aljazzar2011k}. Consequently, our results demonstrating that in Chung-Lu random graphs, most almost shortest nonbacktracking paths are simple strongly suggests that for real-world networks we should identify almost shortest nonbacktracking paths to solve the almost shortest {\em simple} path problem.  

By building upon the work of Byers and Waterman \cite{byers1984determining,matthew2005near}, we provide a simple $O(m\log m+ Lk)$ solution, for finding all (nonbacktracking) paths bounded by length $D$ between two nodes in a directed positively weighted graph, where $L$ is an upperbound for the number of nodes in a path and $k$ is the number of paths returned.   
\newline\newline
An outline of the our paper is as follows:

\begin{itemize}
\item In Section 2 we present a simple algorithm for finding all paths no greater than a prescribed length  in  $O(m\log m+kL)$ time, where $L$ is an upperbound on the number of nodes appearing in any of the $k$ shortest paths.  Furthermore, we illustrate that  for graphs with degree sequences following a power-law distribution, the time complexity of the algorithm is $O(m + kL)$. We also stress that the algorithm works for positively weighted directed and undirected graphs. 
\item Then in Section 3, we introduce the notion of nonbacktracking paths,  In particular, we explore properties of Chung-Lu random graphs in context to the almost shortest simple path problem and prove Corollary \ref{cor:main}, an asymptotic result that demonstrates that for a wide range of parameters, the ratio between the number of simple paths and nonbacktracking paths of prescribed length between two nodes approaches $1$ with high probability.  Subsequently, we illustrate how to extend the algorithm in Section 2 to compute almost shortest nonbacktracking paths with time complexity $O(m\log m + n + kL)$ and space complexity  $O(n +kL)$.  These results strongly suggests that it is often more efficient to use an almost shortest nonbacktracking path algorithm, such as the solution provided in Section 2, to solve for almost shortest simple paths, than an almost shortest simple path algorithm.  
\item And finally in Section 4, we explore applications to the almost shortest paths problem in context to internet routing, where we use our solution for finding almost shortest paths to measure the diversity of surviving paths under an edge deletion process.  We compare (a snapshot of) the AS graph to realizations of an Erdos-Renyi and Chung-Lu random graph and find that the AS graph behaves remarkably similar to realizations of the Chung-Lu random graph model under the edge deletion process.
\end{itemize} 
\section{Almost Shortest Paths Algorithm}
\subsection{Strategy for Finding Almost Shortest Paths} \label{sub:alg}
Even though many algorithms have been proposed for finding almost shortest simple paths between two nodes, very little work has considered the implications for implementing such a solution on real world networks.  We will argue in Section 3 that we should first compute almost shortest {\em nonbacktracking} paths between two nodes to solve the almost shortest simple paths problem.  In this section, we present a simple asymptotically efficient solution for finding {\em all} paths between two nodes less than a certain length.  We will then argue in Section 3 how to extend this algorithm to compute almost shortest nonbacktracking paths. Before presenting the algorithm, we first sketch the solution strategy.

To find all paths between two nodes with length less than a given value, $D$,  our solution constructs a path tree illustrating all possible choices in identifying paths from the source to the target.  As an example, consider finding all paths from node $s$ to node $t$ with length less than $3$ in the graph on the left side of the first panel in Figure \ref{fig:pathspanel}.  We stress that while this example focuses on an undirected unweighted graph, the algorithm will work for directed positively weighted graphs as well.

As mentioned before, informally, the path tree identifies all possible options for constructing almost shortest paths from node $s$ to node $t$.  First, the algorithm maps each node on the path tree to nodes in the original graph.  As all such paths must end with the node $t$, the algorithm maps the root of the tree to the node $t$ in the graph.  The right side of the first panel illustrates this initialization of the path tree.  

Subsequently in the second panel, the algorithm identifies the node(s), $t$, which  corresponds to the node(s) recently added to the path tree in the prior panel; we marked such nodes in blue.  Since all paths from $s$ to $t$ have length at most $3$, it follows that the node that precedes $t$ at the end of the path must have distance at most $2$ from $s$.  Consequently, we mark in yellow the neighbors of the blue node, $t$, with distance at most $2$ from the source. Then the algorithm adds children to the blue node in the tree that correspond to the yellow neighbors of $t$ in the original graph.  In this case, $t$ has three neighbors that have distance $2$ from the node $s$: $a$, $b$ and $c$.

\begin{figure}[h]
\centering
\includegraphics[scale=.6]{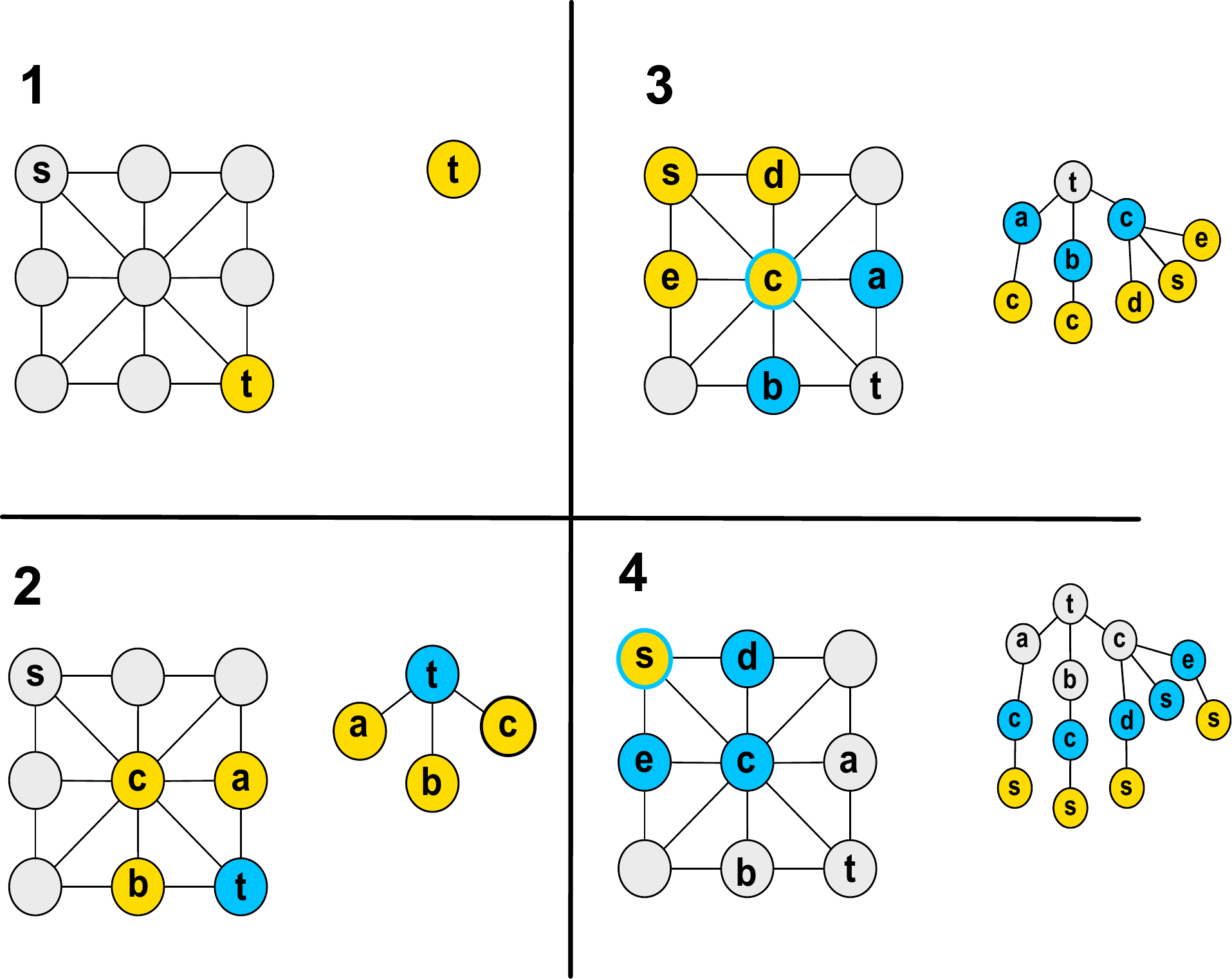}
\caption{An illustration on how to construct a path tree (on the right) for finding paths with length at most $3$ between source $s$ and target $t$ in an unweighted graph (on the left).  Each numbered panel corresponds to a step in constructing the path tree.  In the path tree, yellow nodes indicate newly added nodes, while blue nodes indicate that the nodes were added to the path tree in the prior step.  Similarly, the nodes that correspond to the yellow (blue) nodes in the path tree are also shaded yellow (blue).  Nodes that correspond to both a blue and yellow node in the path tree are shaded yellow with a blue border.}  
\label{fig:pathspanel}
\end{figure}

In Step 3, we repeat the same argument.  For each of the newly added nodes in the tree from the prior step, now marked in blue, we identify that node's neighbors in the original graph with distance at most $1$ from $s$.  Nodes $a$ and $b$ only have one such neighbor, $c$, that satisfies this constraint, so the algorithm adds a child that corresponds to $c$ to those respective blue nodes in the path tree.  Node $c$ in the original graph, which is both a blue and yellow node, has three  neighbors that have distance at most $1$: $d,e$ and $s$.

Finally in Step 4, for each of the newly added nodes in the tree from the prior step, the algorithm identifies the neighbors in the original graph with distance at most $0$ from $s$.   Step 4 completes the construction of the path tree.  We now demonstrate how to efficiently extract all paths from node $s$ to node $t$ with length at most $3$ using the path tree.  See Figure \ref{fig:pathconstruct}.  First, while constructing the path tree, we record all nodes that correspond to $s$ in the original graph. In the first panel, we highlighted all nodes that correspond to $s$ either in green or yellow. We focus on the yellow node.  In the second panel, by looking at the parent of the yellow node in the tree, we can identify the next node on the path, $e$.  Continuing this process for the third and fourth panels yields the path $s,e,c,t$.   

\begin{figure}[h]
\centering
\includegraphics[scale=.6]{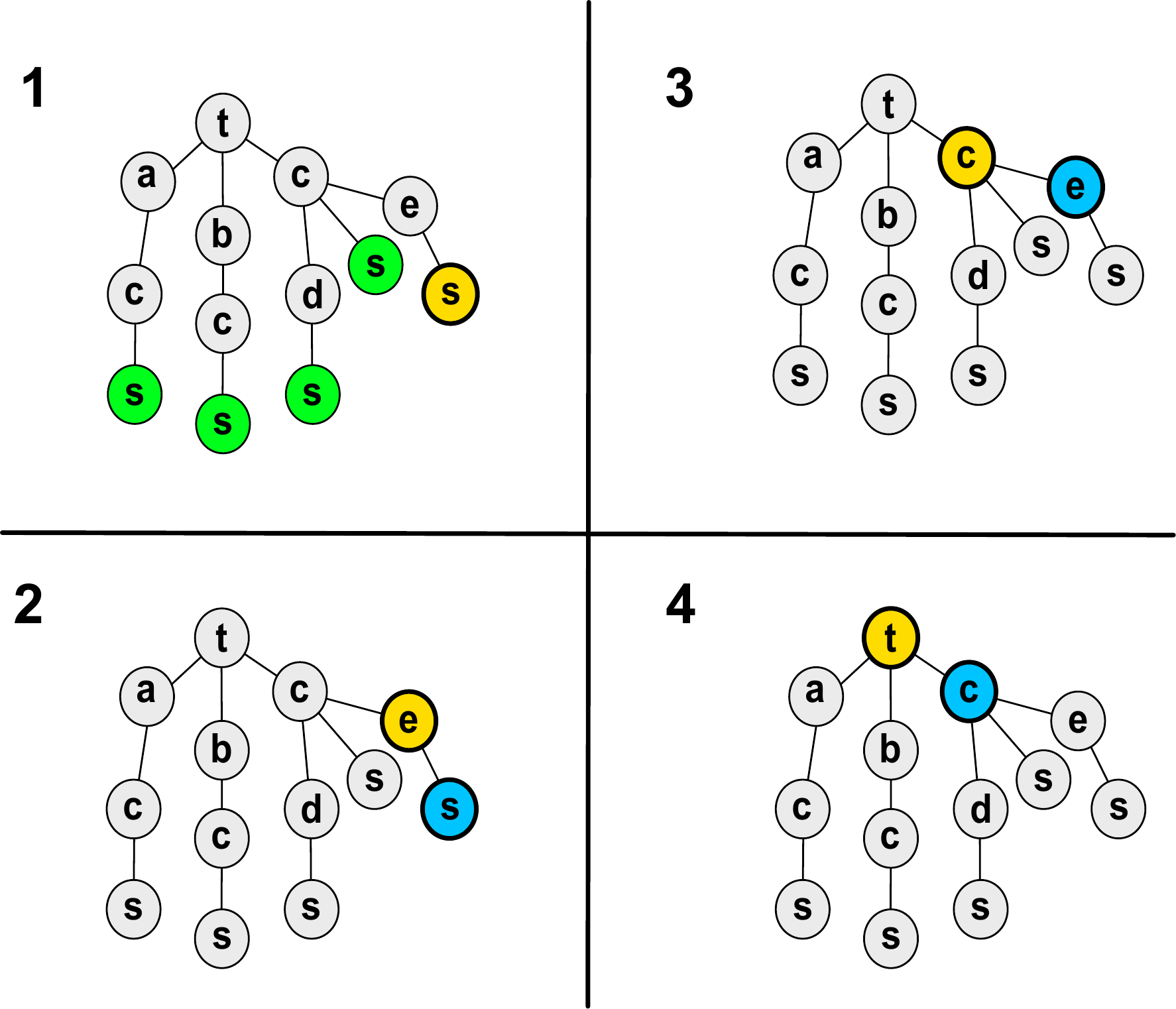}
\caption{An illustration on how to extract paths from the path tree.  We record all nodes that correspond to $s$ while constructing the path tree.  Then for each such node that corresponds to $s$, we traverse the path tree to the root to identify the corresponding path in the original graph.}  
\label{fig:pathconstruct}
\end{figure}

Now that we have explained the intuition behind the proposed solution, at this juncture we specify the inputs (and outputs) for the algorithm, \hyperlink{inputs}{Pathfind}.

\begin{tcolorbox}
\hypertarget{inputs}{\textbf{Pathfind}} requires the following inputs:
\begin{itemize}
\item $V$, a list of nodes in the graph.
\item $InNbrs_{x}$, a list of the incoming neighbors for each node $x \in V$.
\item The \textbf{source} and \textbf{target}, for the almost shortest paths.
\item $d(source,\cdot)$, the distance function from the source to any node in the graph.
\item $d(n,x)$, the positive edge weights in the graph where $n \in InNbrs_{x}$.
\item $D$, the upperbound on the lengths for the almost shortest paths 
\end{itemize}
\textbf{Pathfind} then outputs a list of all paths from the source to the target with length at most $D$.
\end{tcolorbox}

See \hyperlink{alg1}{Algorithm 1} for an outline of the algorithm, \textbf{Pathfind}.  Foremost, to achieve the desired time complexity, Step 1 in \textbf{Pathfind} sorts the incoming neighbors $n$ of each node $x$ according to $d(source,n)+d(n,x)$, the distance from the source to the neighbor $n$ plus the weight of the edge connecting the incoming neighbor $n$ to $x$.  Since nodes can have many neighbors, adding this step will prevent \textbf{Pathfind} from considering a potentially large number of neighbors that are not sufficiently close to the source to form an almost shortest path. 

Subsequently, Step 2 initializes the path tree described in the first panel in Figure \ref{fig:pathspanel}.  Whenever \textbf{Pathfind} adds nodes to the path tree, \textbf{Pathfind} defines the following attributes for such newly added nodes.  The \textbf{Parent} attribute returns the parent of a node on the path tree and the \textbf{ID} attribute of a node in the path tree identifies the corresponding node on the original graph.  Furthermore by construction, edges in the path tree also correspond to edges in the original graph, as neighbors in the path tree correspond to neighbors in the original graph.  Hence, traversing up a given node in the path tree to the root corresponds to a path in the original graph.  Consequently, the \textbf{trackdistance} attribute of a node in the path tree keeps track of the distance of {\em that} path in the original graph. At the conclusion of Step 2  \textbf{Pathfind} also initializes a set  \textbf{Pathstart} to keep track of any nodes in the path tree that correspond to the \textbf{source}, as mentioned in the discussion of Figure \ref{fig:pathconstruct}.   

\begin{figure}[h]
\centering
\includegraphics[scale=.1]{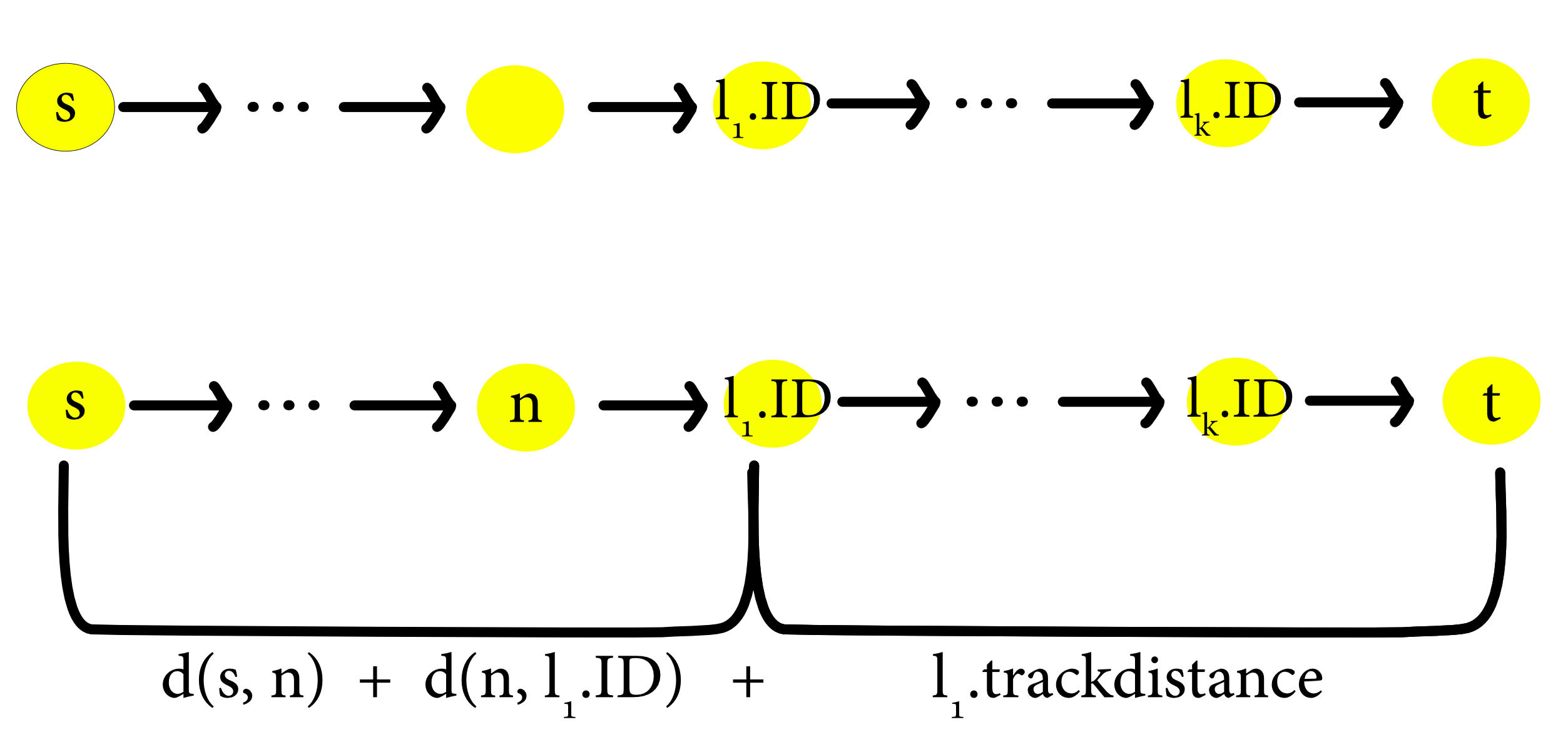}
\caption{(Top) Suppose we know that there exists a path of the form $(s,?,...,?,l_{1}.ID,...,l_{k}.ID,t)$ with length at most $D$.  (Bottom) We then wish to determine if there exists a path of the form $(s,?,...,?,n,l_{1}.ID,...,l_{k}.ID,t)$ with length at most $D$, where $n$ is an incoming neighbor of $l_{1}.ID$.  If we know the length of the path $(l_{1}.ID,...,l_{k}.ID,t)$, is $l_{1}.trackdistance$, then there exists such a path if and only if $d(s, n) + d(n, l_{1}.ID) + l_{1}.trackdistance \leq D$, or equivalently, $d(s, n) + d(n, l_{1}.ID) \leq D - l_{1}.trackdistance$.}
\label{fig:Step3}
\end{figure}

Denote the source and target nodes as $s$ and $t$ respectively.  In Step 3, \textbf{Pathfind} builds the path tree as illustrated in the second to fourth panels in Figure \ref{fig:pathspanel}.    For each recently added node, $l$, to the path tree, \textbf{Pathfind} identifies all  neighbors $n$ of $l.ID$ in the original graph such that $d(s,n)+d(n,l.ID)$ is {\em sufficiently small} as illustrated in Figure \ref{fig:Step3}.   More precisely,  suppose we know that there exists a path with length at most $D$ of the form $(s,?,...,?,l_{1}.ID,...,l_{k}.ID,t)$, where the nodes $s, l_{1}.ID,...,l_{k}.ID, t$ are fixed and there are no constraints on the nodes connecting $s$ to $l_{1}.ID$.  (In practice, they are unknown.)  We then wish to determine if there exists a path with length at most $D$ of the form $(s,?,...,?,n,l_{1}.ID,...,l_{k}.ID,t)$, where $n$ is an incoming neighbor of $l_{1}.ID$.  Furthermore, suppose that we record the length of the path $(l_{1}.ID,...,l_{k}.ID,t)$ as $l_{1}.trackdistance$.  Consequently, if there exists a path of the form $(s,?,...,?,n,l_{1}.ID,...,l_{k}.ID,t)$ with length at most $D$, it follows that $d(s, n) +d(n, l_{1}.ID) + l_{1}.trackdistance \leq D$, or equivalently 
\begin{equation} \label{eq:small} d(s,n) + d(n, l_{1}.ID) \leq D - l_{1}.trackdistance. 
\end{equation}
For each such neighbor $n$ that satisfies the above inequality, \textbf{Pathfind} adds a new node to the path tree, defining the attributes, \textbf{Parent},  \textbf{ID} and \textbf{trackdistance} accordingly.  In the event the neighbor $n$ under consideration is the \textbf{source}, \textbf{Pathfind} adds the corresponding node in the path tree to the set  \textbf{Pathstart}.

\noindent Finally, Step 4 extracts paths from the path tree, as described in Figure \ref{fig:pathconstruct}.  That is for each $treenode$, where $treenode.ID$ is the \textbf{source}, \textbf{Pathfind} traverses up the tree to the root to construct an explicit representation of a path from the  \textbf{source} to the \textbf{target}.  \textbf{Pathfind} then returns all such paths with length bounded by the presrcribed parameter $D$.
\begin{tcolorbox}
\hypertarget{alg1}{\textbf{Algorithm 1: Pathfind}}
\begin{enumerate}
\item (Sort adjacency list). For each node $x \in  V$, sort the nodes  $n \in InNbrs_{x}$ by $d(source,n)+d(n,x)$ in non-decreasing order.    We can then easily identify incoming neighbors of $x$ that are close to the source.
\item (Initialize path tree.) Initialize a path tree with the single node \textbf{root}.  We will construct a correspondence between paths on the path tree to almost shortest paths between the source and target in the original graph.
\begin{enumerate}
\item Define the following attributes for nodes in the path tree.  
\begin{enumerate}
\item \textbf{Parent} returns the parent of a node on the tree. Set $root.Parent$ $\gets \emptyset$. 
\item \textbf{ID} maps nodes in the tree to nodes in the graph.  Set $root.ID$ $\gets target$.
\item \textbf{trackdistance} tracks the distance traveled so far in the original graph.  
Set $root.trackdistance$ $\gets 0$.  
\end{enumerate} 
\item Initialize a set \textbf{PathStart} $\gets\emptyset$, where PathStart will contain all nodes $t$ in the path tree, such that $t.ID$ is the source.
\item Initialize a queue, $Q$, with node root. $Q$ will help us identify almost shortest paths between the source and target.
\end{enumerate}
\item (Search for almost shortest paths by constructing path tree). While $Q \neq \emptyset$, remove the first element $l$ from $Q$.
\begin{enumerate}
\item For each $n \in InNbrs_{l.ID}$, check if there exists a path from source to $l.ID$, where the last edge connects $n$ to $l.ID$ and the path has length at most $ D - l.trackdistance$.  If at any point we find an $n$ such that there does not exist such a path, we should exit this for loop as all remaining nodes in $InNbrs_{l.ID}$ are too far away from the source by Step 1.  Else, we should add a new node $z$ to the tree with the following attributes.

\begin{enumerate} 
\item $z.parent \gets l$
\item $z.ID \gets n$
\item $z.trackdistance \gets d(n,l.ID) + l.trackdistance$ and
\item $Q \gets Q \cup z$.
\item If $n == source$ then add $z$ to $PathStart$.
\end{enumerate}

\end{enumerate}
\item (Construct explicit representation for almost shortest paths). Initialize the output, the list of almost shortest paths, $PathList \gets \emptyset$.   Then for each $v \in PathStart$, find the shortest path from $v$ to the root in the tree path by using the parent attribute as mentioned in Figure \ref{fig:pathconstruct}.  
\begin{enumerate}
\item Denote the path in the path tree as $(v,v_{1}....,v_{k})$, where $v_{k}$ is the root.  Then the corresponding path in the original graph is $(v.ID,v_{1}.ID,....,v_{k}.ID)$.  Note that by construction, $v.ID$ is the source and $v_{k}.ID = root.ID$ is the target.  
\item $PathList \gets PathList \cup (v.ID,v_{1}.ID,....,v_{k}.ID)$.
\end{enumerate}
\item Return $PathList$.
\end{enumerate}
\end{tcolorbox}

\subsection{Complexity Analysis}
We now verify the claimed computational complexity of the algorithm for identifying all paths between two fixed nodes of length bounded by $L$. 

\begin{thm}
Denote $d_{in}(x)$ as the in-degree of node $x$, let $V$ be the set of all nodes in the graph and $m$ be the total number of edges in the graph.   If $\sum_{x \in V} d_{in}(x)\log d_{in}(x) = O(m)$, then the time complexity for the \hyperlink{alg1}{Pathfind} algorithm (in section \ref{sub:alg}) is $O(m+kL)$, where $k$ is the number of shortest paths in the output, $n$ is the number of nodes in $V$ and $L$ is an upperbound for the number of nodes in any outputted path.  Otherwise, the computational complexity is $O(m\log m+n+kL)$
\end{thm}

\begin{proof}

Step 1 in \hyperlink{alg1}{Pathfind} sorts the incoming neighbors for each node $x$, $InNbrs_{x}$.  Denote $d_{in}(x)$ as the number of incoming neighbors for node $x$ and let $V$ be the set of all vertices.  By using a heapsort, \hyperlink{alg1}{Pathfind} can sort the sets $InNbrs_{x}$ for all $x$ in  $\sum_{x \in V} O(d_{in}(x)\log d_{in}(x))$ time.  This quantity is trivially bounded by $\sum_{x \in V} O(d_{in}(x)\log m) = O(m\log m)$, as the maximum in-degree is bounded by the number of edges.  Alternatively, for many real world networks,  $\sum_{x \in V}d_{in}(x)\log d_{in}(x) = O(m)$.  Hence it would follow that if either the graph is nicely weighted or if $\sum_{x \in V}d_{in}(x)\log d_{in}(x) = O(m)$, then Step 1 takes $O(m)$ time.  Otherwise, Step 1 takes $O(m \log m )$ time. 

Step 2 takes constant time $O(1)$ as Pathfind initializes the path tree.

For Step 3, note that the complexity for evaluating the criteria to determine whether we add a node to our path tree is $O(1)$.  Furthermore, for every neighbor of $l.ID$ that satisfies the criteria must yield at least one path.  And since $InNbrs_{l.ID}$ is sorted, there is only an $O(1)$ penalty when Pathfind comes across a neighbor that is not sufficiently close to node $s$ in the graph as all other neighbors that  have not been checked are too far to construct an almost shortest path.  Since the time complexity for adding a node or edge to the path tree is $O(1)$, \textbf{the complexity of Step 3 is proportional to the number of nodes and edges in the tree}.  

We claim that the number of nodes and edges in the tree is $O(kL)$.  As illustrated in panel 4 in Figure \ref{fig:pathspanel} and in Figure \ref{fig:pathconstruct}, all leafs $z$ in the path tree have the property that $z.ID$ maps to node $s$ in the original graph.  Similarly, the root of the tree corresponds to node $t$ in the original graph.  In particular, for any leaf, the shortest path from that leaf to the root corresponds to an outputted almost shortest path from $s$ to $t$.  

Since we assumed that outputted paths from $s$ to $t$ contain at most $L$ nodes, then there are at most $L$ nodes on the shortest path between any leaf and the root.  Note that every node in the path tree appears in at least one shortest path between a leaf and the root.  As each leaf in the tree corresponds to a different outputted almost shortest path and Pathfind outputs $k$ paths, there are at most $k$ leaves in the tree.  Consequently, the number of nodes in the path tree is $O(kL)$. Furthermore, since the number of edges in a tree equals the number of nodes minus one, we conclude that the number of edges is also $O(kL)$.  Hence the complexity of Step 3 is $O(kL)$.

Finally in Step 4, the time to reconstruct explicit representations for each path is $O(L)$ as the shortest path from any node to the root contains at most $L$ nodes by assumption.  And since Pathfind outputs at most $k$ paths this step also costs $O(kL)$.  The combined computational complexity of the algorithm is therefore $O(m + kL)$ or $O(m\log m + kL)$ depending on the assumptions from Step 1.  

\end{proof}  

\noindent \textbf{Remark:} Many real world networks exhibit degree sequences that appear to follow a scale free distribution \cite{wang2003complex,chakrabarti2006graph,barabasi2009scale}, that is $Pr(d_{in}(x) = r) \propto r^{-\gamma}$, for some positive value of $\gamma$. (Typically $\gamma$ is greater than 2.).  Let $n$ be the number of nodes in the network.  If for instance $\gamma > 2$, it follows that, \begin{equation} \sum_{x \in V} d_{in}(x)\log d_{in}(x) \approx n\int_{1}^{n}Pr(d_{in}(x) = r)r\log (r) dr\propto n\int_{1}^{n}r^{-\gamma + 1}\log (r)dr, \end{equation}

where $\int_{1}^{n}Pr(d_{in}(x) = r)r\log (r) dr$ is roughly the expected value for the in-degree times the logarithm of the in-degree. Notice that we integrate up to $n$ as the in-degree of a node cannot exceed the number of nodes in the network.  

Using integration by parts it follows that  $$n\int_{1}^{n}r^{-\gamma + 1}\log (r) dk = O(\frac{n}{(\gamma - 2)^{2}}).$$  For networks that exhibit a scale free distribution with $\gamma >2$ and $ n < m$, we conclude that $$\sum_{x \in V}d_{in}(x)\log d_{in}(x) = O(m).$$ 

Now that we have verified that the time complexity for Pathfind, we now consider the space complexity.
\begin{lem}
The space complexity for Pathfind is $O(kL)$.
\end{lem}
\begin{proof}
From Step 1, sorting takes $O(1)$ additional space. From Steps 2-5, constructing the path tree takes $O(kL)$ space as there are at most $kL$ nodes in the tree.  Consequently, Pathfind has $O(kL)$ space complexity. 
\end{proof}

With the space and time complexity results for Pathfind at hand, we verify that Pathfind does indeed find all paths of length bounded by $L$.

\begin{thm}
\hyperlink{alg1}{Pathfind} finds all paths of length bounded by $D$ from node $s$ to node $t$.
\end{thm}

\begin{proof}
By construction of the algorithm if two nodes $x$ and $y$ in the path tree are neighbors, then $x.ID$ and $y.ID$ are neighbors in the original graph $G$.  Consequently, it follows that for any path $P = (v_{1},...,v_{k})$ in the path tree, where $v_{1}.ID$ is the source and $v_{k}.ID$ is the target, then $(v_{1}.ID,...,v_{k}.ID)$ is a path from the source, s, to the target, t. Hence, all paths returned by Pathfind are paths from $s$ to $t$.

Alternatively for any path in the original graph $(n_{1},....,n_{k})$, with length bounded by $D$, where $n_{1}=s$, the source  and $n_{k}=t$, the target, we need to show that there is a corresponding path in the path tree.  Inductively starting with the node $n_{k}$, let $v_{k}$ be the root of the tree and it follows that  $v_{k}.ID = n_{k}$. Now by construction, since $d(n_{k-1},n_{k}) + d(n_{1},n_{k-1}) \leq D$, as $(n_{1},....,n_{k})$ has length bounded above by $D$, it follows that there is a child $v_{k-1}$ of the root ($v_{k}$), where $v_{k-1}.ID = n_{k-1}$. 

Furthermore by construction, it follows that there is a unique child of $v_{k-1}$, $v_{k-2}$, where $v_{k-2}.ID = n_{k-2}$, as $d(n_{1},n_{k-2}) + d(n_{k-1},n_{k-2}) \leq D - d(n_{k-1},n_{k})$, where $d(n_{k-1},n_{k}) = v_{k-1}.trackdistance$.  Proceeding inductively, we conclude that for all $j \in \{1,...,k\}$, there is a node $v_{j}$ in the path graph with the properties that if $j < k$, $v_{j}.Parent = v_{j+1}$ and $v_{j}.ID = n_{j}$.  Furthermore, since $v_{1}.ID = s$, the source, this implies that $v_{1} \in PathStart$.   Consequently, there is a bijection between the paths from $s$ to $t$ with length at most $D$ and the paths from nodes in $PathStart$ to the root of the path tree. 
\end{proof}

\section{The Ratio of Simple to Nonsimple Paths in Chung-Lu Random Graphs}
Intuitively, since real world networks are locally tree like, to efficiently identify almost shortest {\em{simple}} paths, we should use an {\em almost shortest path} algorithm.  Consequently, we employ the Chung-Lu random graph model as a convenient method for constructing a collection of graphs that emulate properties of real world networks.  To this end, we seek conditions for realizations of the Chung-Lu random graph model, where the number of simple paths of fixed length is roughly the same as the number of paths of that length.  Unfortunately, for many {\em{undirected} }random graphs with nodes of large degree, the aforementioned claim is false \cite{castellano2010thresholds}; short non-simple paths in undirected graphs can outnumber simple paths by considering paths that revisit nodes of large degree.  To circumvent this issue, we introduce nonbacktracking paths,where paths cannot traverse the same edge twice in a row.  After illustrating that under a broad range of parameters that the number of nonbacktracking paths asymptotically approximates the number of simple paths of the same length (Corollary \ref{cor:main}), we then demonstrate how to adapt the almost shortest paths algorithm in the prior section to compute the almost shortest nonbacktracking paths with the same computational complexity, if the number of edges $m$ exceeds the number of nodes (Lemma \ref{lem:QSPcomp}).

\begin{defn}\textbf{Chung-Lu Random Graph Model \cite{chung2002average}}: Let $n$ be the number of nodes in an undirected graph and let $\mathbf{d} = (d_{1},...,d_{n})$ be the expected degree sequence, where $d_{i}$ corresponds to the expected degree of node $i$. Denote $S = \sum_{i=1}^{n}d_{i}$ and suppose that $\max_{i} d_{i}^{2} \leq S$.  We then model edges in the graph as independent Bernoulli random variables.  In particular, we denote the probability an edge exists connecting nodes $i$ and $j$ as $p_{ij}$, where $p_{ij}=\frac{d_{i}d_{j}}{S}$.  
\end{defn}

As a technical point in the above definition, nodes may have edges that connect to themselves. But before introducing any subsequent results regarding the Chung-Lu random graph model, the following notation will be helpful.
\begin{defn}Define the random variable, $SP_{r}(s,t)$, to be the number of simple paths from node $s$ to $t$ with length $r$.  
\end{defn}

To calculate the number of simple paths in the graph, we employ Hoare-Ramshaw notation for a closed set of integers, namely $$[a..b] = \{x\in \mathbb{Z}: a\leq x \leq b\}.$$  At this juncture, we show that in expectation the number of simple paths between any two nodes grows exponentially for Chung-Lu random graphs.  
\begin{lem}\label{lem:exp}
For the Chung-Lu random graph model, define $S_2 = \sum_{i=1}^{n} d_{i}^{2}$ and $d_{max} = \max_{i} d_{i}$.   Consider the expected number of simple paths of length $r$ from node $s$ to $t$, $E[SP_{r}(s,t)]$.  Furthermore, let $p_{max} = \frac{d_{max}^{2}}{S}$.  It then follows that

$$p_{st}\frac{S_2}{S}^{r-1}(1-\frac{r(r+1)p_{max}}{2}\frac{S}{S_2}) \leq E[SP_{r}(s,t)] $$

where $p_{st}$ is the probability that an edge exists connecting the nodes $s$ and $t$.  
\end{lem}

\begin{proof}   Define the set $B_{st}$ such that $\mathbf{b} \in B_{st}$ if $\mathbf{b} = (s,b_{1},...,b_{r-1},t)\in \mathbb{N}^{r+1}$, where each entry in $\mathbf{b}$ is distinct and  $b_{i}\in [1..n]$ for all $i\in[1..r-1]$.  Informally, $B_{st}$ consists of all simple paths from $s$ to $t$ of length $r$ that could exist in a graph of $n$ nodes.  It then follows that the expected the number of  simple paths of length $r$ between nodes $s$ and $t$ is the sum of probabilities that a given simple path from $s$ to $t$ exists. 

$$E[SP_{r}(s,t)] =\sum_{\mathbf{b}\in B_{st}} p_{sb_1}(\Pi_{k=1}^{r-2}p_{b_{k}b_{k+1}})p_{b_{r-1}t},$$

Using the probabilities that two nodes share an edge in the Chung-Lu random graph model, we can rewrite the above expression.

\begin{equation} E[SP_{r}(s,t)] =\sum_{\mathbf{b}\in B_{st}} \frac{d_{s}d_{b_1}}{S}(\Pi_{k=1}^{r-2}\frac{d_{b_k}d_{b_{k+1}}}{S})\frac{d_{b_{r-1}}d_{t}}{S}
\end{equation}

Noticing that for each index $i$ from $1$ to $r-1$, the term $d_{b_{i}}$ appears twice in the product, we will argue that the following inequality holds.

\begin{multline} E[SP_{r}(s,t)] =  \frac{d_{s}d_{t}}{S}\sum_{\mathbf{b}\in B_{st}}\Pi_{k=1}^{r-1}\frac{d_{b_k}^{2}}{S}\geq \\ 
p_{st}(\sum_{b_1=1,...,b_{r-1}=1}^{n}\Pi_{k=1}^{r-1}\frac{d_{b_k}^{2}}{S} - \binom{r+1}{2}\sum_{\substack{b_2=1,...,b_{r-1}=1}}^{n}\frac{d_{max}^{2}}{S}\Pi_{k=2}^{r-2}\frac{d_{b_k}^{2}}{S}),
\end{multline}

\noindent where $d_{max} = \max \mathbf{d}$ and we derive the last inequality by inclusion-exclusion; we consider the contribution of the summation by removing the constraint of the distinctness of the $b_{i}$ terms and then we subtract off terms where the $b_{i}$ either equal $s$, $t$, or another $b_{j}$.
To compute the quantity we should substract off, we first consider the contribution where $b_{1} = b_{2}$ and then multiply that quantity by $\binom{r+1}{2}$, corresponding to the number of ways any two of the $r+1$ nodes in the path could equal each other and hence violate the constraints in the original summation. It then follows that 

$$E[SP_{r}(s,t)] \geq p_{st}(\frac{S_2}{S}^{r-1} - \binom{r+1}{2}\sum_{\substack{b_2,...,b_{r-2}}}p_{max}\Pi_{k=2}^{r-2}\frac{d_{b_k}^{2}}{S})\geq$$ $$p_{st}(\frac{S_2}{S}^{r-1} - \binom{r+1}{2}p_{max}\frac{S_2}{S}^{r-2})=p_{st}\frac{S_2}{S}^{r-1}(1-\frac{r(r+1)p_{max}}{2}\frac{S}{S_2}).$$

\end{proof}
As a result of Lemma \ref{lem:exp}, when calculating the $k$ shortest paths, the expected number of simple paths grows exponentially in terms of length.  Consequently, we may be arbitrarily or perhaps even systematically ignoring many paths of the same length.  For this reason in many applications, it is often more informative to calculate all paths bounded by a fixed length as opposed to calculating just $k$ of them.  Since we wish to show that the ratio between the number of simple paths and (nonbacktracking) paths of the same length is well behaved, we seek bounds for the expected number of (nonbacktracking) paths of length $r$ between two arbitrary nodes; however, since the same edge may appear multiple times on a path, we first provide an efficient method for computing the probability that an arbitrary path exists.  To do so, we will need the following definitions.

\begin{defn}
Consider an edge in a given path.  If the edge has not appeared before, that edge is a \textbf{new edge}.  Alternatively, if the edge has appeared before, that edge is a \textbf{repeating edge}.  Furthermore, a list of consecutive repeating  edges of maximal size in a path is called a \textbf{repeating edge block}.  We can define a \textbf{new edge block} similarly, as a list of consecutive new edges of maximal size in a path.  The \textbf{new edge interior}  is a list of nodes that includes the mth node in the path if the incoming edge to the mth node and the outgoing edge from the mth node are new edges. Note that the new edge interior excludes the first and last nodes in the path. 
\end{defn}
In order to construct a convenient formula for computing the probability that a given path exists, it will be helpful to identify which nodes appear elsewhere in the path.  
\begin{lem}\label{lem:pathrestrict}
Let $x$ be a node in an undirected graph that appears in a repeating edge block of a path and is not the first node in that repeating edge block.  Then one of the following must be true:
\begin{itemize}
\item $x$ must also be the first node in the path
\item $x$ must also appear in the new edge interior
\item or $x$ appears earlier in the path as the first node of a repeating edge block.
\end{itemize}
\end{lem}
\begin{proof}
Suppose that there exists a path $P$ that contradicts the lemma.  In particular, consider the first node, $x$ in the path, that contradicts the lemma statement. Suppose that this first contradiction appears as the mth node in the path. As $x$ is not the first node in the repeating edge block, we know that there exists an edge of the form $(y,x)$ in the repeating edge block for some node $y$.  By the definition of a repeating edge, $(y,x)$ or $(x,y)$ appears earlier as a new edge in the path.

\textbf{Case 1:} $(y,x)$ appears earlier as a new edge in the path.  If $(y,x)$ is a new edge in the path, then either $x$ is in the interior of a new edge block or $x$ is at the end of a new edge block and hence would also be the first node of a repeating edge block. 

\textbf{Case 2:} $(x,y)$ appears earlier as a new edge in the path.  Then $x$ can either be the first node in the path, $x$ can be in the new edge interior, or $x$ is at the beginning of a new edge block.  If $x$ is at the beginning of a new edge block, then $x$ is at the end of a repeating edge block. This implies that this earlier appearance of $x$ also contradicts the lemma.  But since we stipulated that the first contradiction must appear at the mth node in the path and not earlier, the lemma must be true. 
\end{proof}

At this juncture, we present notation to perform summations over nodes in a path (or entries in a list).  In particular, we can treat lists as sets; formally, we can represent a list containing natural numbers  as a function $f: \mathbb{N}\rightarrow\mathbb{N}$.  We can then represent this function as a set of ordered pairs.  More specifically, if $(4,7)$ is in the set corresponding to this list, then the $4th$ entry in the list is $7$.  In this way, we can carry over set notation to lists.  For example for a list $\mathbf{N}$, we can say that $(4,7)\in \mathbf{N}$.  Unfortunately, it is notationally burdensome and frequently uninsightful to refer to the position of the entry on a list and so we omit this information.

With this method, for a list $\mathbf{N}$ and an arbitrary function $g:\mathbb{N}\rightarrow\mathbb{R}$, we can define $\Pi_{j\in\mathbf{N}}g(j) = \Pi_{(i,j)\in\mathbf{N}}g(j)$.  In particular if $\mathbf{N} = [3,5,5]$.  Then $\Pi_{j\in\mathbf{N}}g(j) = g(3)\cdot g(5) \cdot g(5)$.

We now provide the following result, which will assist us in computing the probability that a path exists even if multiple edges repeat.
\begin{lem}\label{lem:probpath}
Define $X_{(i,j)}$ as an indicator random variable that equals $1$ if the edge $(i,j)$ exists and $0$ otherwise.  Consequently, $\Pi_{k=1}^{r}X_{(i_{k},i_{k+1})}$ is an indicator random variable that equals $1$ if there is a path $(i_{1},...,i_{r+1})$.  Let $\mathbf{N}$ be the list of all nodes in the new edge interior.  Let $\mathbf{R_{1}}$ be a list of pairs of the first and last nodes for each repeating edge block, where the first node has appeared before and let $\mathbf{R_{2}}$ be a list of pairs of the first and last nodes for each repeating edge block, where the first node has not appeared before.  If the first and last edges are new edges, then 
\begin{equation}\label{eq:prob}
Pr(\Pi_{k=1}^{r}X_{(i_{k},i_{k+1})} = 1) = \frac{d_{i_{1}}d_{i_{r+1}}}{S}\Pi_{i\in\mathbf{N}}\frac{d_{i}^{2}}{S}\Pi_{(j,k)\in \mathbf{R_{1}}}\frac{d_{j}d_{k}}{S}\Pi_{(l,m)\in \mathbf{R_{2}}}\frac{d_{l}d_{m}}{S}.
\end{equation}

Furthermore, if $q_{i}$ is the number of repeating edge blocks of length $i$, then the number of nodes in $\mathbf{N}$, \begin{equation}\label{eq:count} |\mathbf{N}| = r - 1 - \sum_{i=1}^{r-2}(i+1)q_{i}. 
\end{equation}
\end{lem}
\begin{proof}
To derive (\ref{eq:prob}), we first consider some examples from Figure \ref{fig:repeatedge}.
\begin{figure}[h]
\centering
\includegraphics[scale=.1]{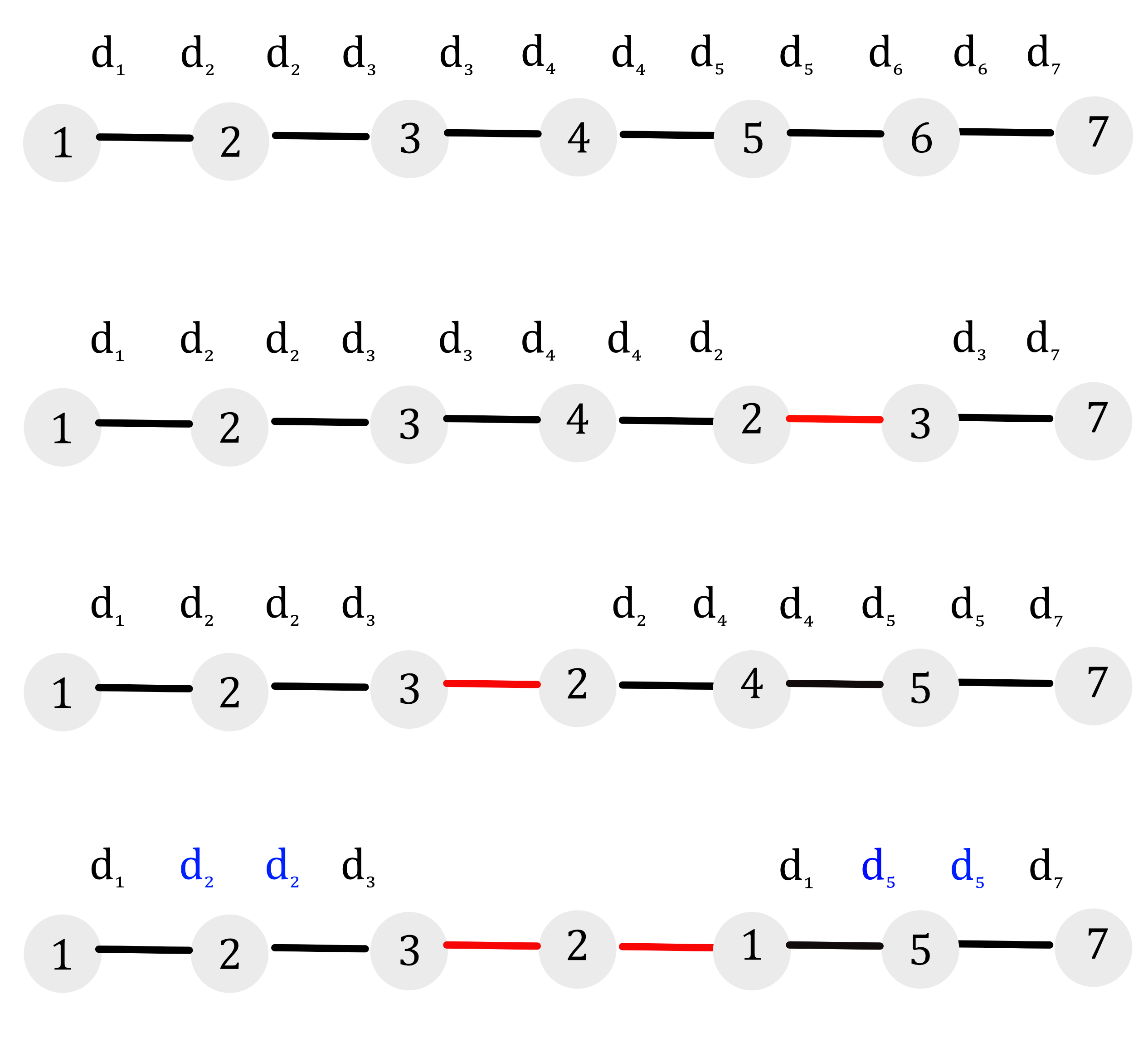}
\caption{An illustration for the types of repeating edge blocks.  For each of the four paths, edges shaded in red represent repeating edges.  \textbf{For the first path}, there are no repeating edges.  Above each non-repeating edge we write the product of the respective expected degrees of two nodes, which is proportional to the probability the edge exists. \textbf{In the second path}, we have one repeating edge, where the first node in that repeating edge block appears elsewhere in the path.  In the notation of Lemma \ref{lem:probpath}, $\mathbf{R}_1 = [(2,3)]$.  \textbf{For the third path}, the first node in the repeating edge block does not appear elsewhere. $\mathbf{R}_2 = [(3,2)]$, $\mathbf{R}_1 = \emptyset$.  \textbf{For the fourth path}, we illustrate how a larger repeating edge block impacts the probability that a path exists.  $\mathbf{R}_2 = [(3,1)]$, $\mathbf{N} = [2,5]$.} \label{fig:repeatedge}
\end{figure}
For the first path in Figure \ref{fig:repeatedge}, there are only new edges and the new edge interior $\mathbf{N} = [2,3,4,5,6]$.   It follows that the probability that the first path in Figure \ref{fig:repeatedge} exists is precisely $$\frac{d_{1}d_{7}}{S}\Pi_{i=2}^{6}\frac{d_{i}^{2}}{S} = \frac{d_{1}d_{7}}{S}\Pi_{i\in \mathbf{N}}\frac{d_{i}^{2}}{S}.$$

More generally for an arbitrary path from node $s$ to node $t$ with no repeating edges, it follows that the probability the path exists is,

$$\frac{d_{s}d_{t}}{S}\Pi_{i\in \mathbf{N}}\frac{d_{i}^{2}}{S}.$$

Of course as in the second path in Figure \ref{fig:repeatedge}, we may have repeating edges. 
By noting that the last node in a repeating edge block must appear earlier (by Lemma \ref{lem:pathrestrict}), there are two types of repeating edge blocks; either the first node has appeared earlier in the path or the first node has not appeared earlier in the path.

Define a list $\mathbf{R_1}$ consisting of the first and last nodes for each repeating edge block, where the first node has appeared earlier in the path.  Consequently, the probability such a path exists is,

$$\frac{d_{s}d_{t}}{S}\Pi_{i\in \mathbf{N}}\frac{d_{i}^{2}}{S}\Pi_{(j,k)\in\mathbf{R_1}}\frac{d_{j}d_{k}}{S}.$$

Finally, since the first node in a repeating edge block may have not appeared before, define a list $\mathbf{R_{2}}$ consisting of the the first and last nodes for each repeating edge block, where the first node has not been seen before.  Then the probability such a path exists is

\begin{equation}\label{eq:comp}\frac{d_{s}d_{t}}{S}\Pi_{i\in \mathbf{N}}\frac{d_{i}^{2}}{S}\Pi_{ (j,k)\in\mathbf{ R_1}}\frac{d_{j}d_{k}}{S}\Pi_{(l,m)\in \mathbf{R_2}}\frac{d_{l}d_{m}}{S}.
\end{equation}

\noindent This completes the proof of (\ref{eq:prob}).   Let $q_{i}$ be the number of repeating edge blocks of length $i$.  We now verify equation (\ref{eq:count}), $|\mathbf{N}| = r - 1 - \sum_{i=1}^{r-2}(i+1)q_{i}.$  
\newline\newline \noindent First consider the number of times a node is not in the new edge interior, denoted by $|\mathbf{N}^{c}|$.  Alternatively, $|\mathbf{N}^{c}|$ counts the number of times a node appears in a repeating edge block in addition to the first and last nodes of the path.   Then $|\mathbf{N}^{c}|$ is precisely $2 + \sum_{i}(i+1)q_{i}$, as there are $i+1$  nodes in each of the $q_{i}$ repeating edge block of length $i$.  Since there are $r + 1$ nodes in a path of length $r$, $r + 1 - 2 - \sum_{i\geq 1}(i+1)q_{i}$ is precisely the right hand side of equation (\ref{eq:count}).  (We derive the upperlimit in the summation from the assumption that the first and last edges are new edges and that the path has length $r$, which implies that the length of a repeating edge block could be at most $r-2$.) 
\end{proof}

 At this juncture, we provide a formal definition for a nonbacktrackingpath; we will then show that such paths are both analytically tractable and easy to compute using an almost shortest path algorithm.  

\begin{defn}
A nonbacktracking path $(x_{1},...,x_{r+1})$ is a path where for all integers $i\in[1..r-1]$, $x_{i}\neq x_{i+2}$.  We denote the number of nonbacktracking paths of length $r$ from node $s$ to $t$ as $NBP_{r}(s,t)$.
\end{defn}

In the following lemma, we demonstrate how the formula from Lemma \ref{lem:probpath} simplifies for computing the probability that a nonbacktracking path exists.  

\begin{lem}
Given a nonbacktracking path, then the first node for any repeating edge block in the path either appears in the new edge interior or is the first node in the path.  Alternatively in the language of Lemma \ref{lem:probpath} for a nonbacktracking path, $\mathbf{R}_2 =\emptyset$.
\end{lem}
\begin{proof}
Suppose there exists a nonbacktracking path $(x_{1},...,x_{r+1})$ that violates the lemma and denote the first node in the path, $x_{i}$, that is the first node in a repeating edge block that does not appear in the new edge interior and is not the first node in the path.  It follows that $(x_{i-1},x_{i})$ is a new edge and that $(x_{i},x_{i+1})$ has appeared elsewhere.  Once we prove that $x_{i}$ cannot appear earlier in the path, it will follow that $x_{i+1} = x_{i-1}$ and that the path will not be a nonbacktracking path, a contradiction.
\newline
Suppose that $x_{i}$ has appeared earlier in the path.  $x_{i}$ cannot be the first node in the path or part of the new edge interior. Consequently, there exists a $j$ such that $x_{j} = x_{i}$, where $j< i $; furthermore, $x_{j}$ must appear in an earlier repeating edge block.  From Lemma \ref{lem:pathrestrict}, this would imply that there exists an $x_{k} = x_{j} = x_{i}$, where $x_{k}$ appears at the {\em beginning} of an earlier repeating edge block. But then this would imply that the $ith$ node in the path, $x_{i}$ is not the {\em first} node in the path that violates the property stated in the lemma.  Consequently, $x_{i}$ cannot appear earlier in the path.  
\end{proof}

Now that we have demonstrated that for a nonbacktracking path, the first node in a repeating edge block appears either in the interior of a new edge block or is the first node in a path, we can invoke Lemma \ref{lem:probpath} to bound the expected number of nonbacktracking paths between two nodes.

\begin{thm}\label{thm:Epaths}
For the Chung-Lu random graph model, define $S_2 = \sum_{i=1}^{n} d_{i}^{2}$ and $d_{max} = \max_{i} d_{i}$.  Let $S_2 > S$ and consider the expected number of nonbacktracking paths of length $r$ from node $s$ to $t$, $E[NBP_{r}(s,t)]$, where $s \neq t$.  If $2r < \frac{S_2}{S}$ , then

$$E[NBP_{r}(s,t)] \leq \frac{p_{st}(\frac{S_2}{S})^{r-1}}{1-\frac{S}{S_2}}exp(\frac{(2r\frac{S}{S_2})^{2}p_{max}}{1-\frac{2rS}{S_2}}).$$
\end{thm}
\begin{proof}
The challenge in bounding the expected number of nonbacktracking paths of length $r$ comes from the issue that a path may visit the same edge multiple times.  As a result, we define the indicator random variable $X_{(u,v)}$ to be $1$ if the edge $(u,v)$ exists and $0$ otherwise.  For simplicity let $i_0 = s$ and $i_r = t$. Define a set $\mathbf{B}\subset \mathbb{N}^{r-1}$, where $\mathbf{i} \in \mathbf{B}$ if for each $j\in [1..r-1]$, $i_{j} \in [1..r]$ and for all $j \in [2..r]$, $i_{j} \neq i_{j-2}$.  (Alternatively, if $\mathbf{i} \in \mathbf{B}$, then $\mathbf{i}$ corresponds to a nonbacktracking path from $s$ to $t$ that could exist in the graph.)    We then have that

\begin{equation} \label{eq:indicateE}
E[NBP_{r}(s,t)] =\sum_{\mathbf{i}\in\mathbf{B}}Pr(\Pi_{j=0}^{r-1} X_{(i_{j},i_{j+1})} = 1) = 
\end{equation}
\begin{equation} \sum_{\mathbf{i}\in\mathbf{B}}Pr(X_{(i_{0},i_{1})}=1)\Pi_{j=1}^{r-1}\Pr(X_{(i_{j},i_{j+1})} = 1|\Pi_{k=0}^{j-1}X_{(i_{k},i_{k+1})}= 1),
\end{equation}

\noindent where $\Pi_{j=0}^{r-1} X_{(i_{j},i_{j+1})} = 1$ implies that there is a path (of length $r$) from $i_{0} = s$ to $i_{r} = t$.  Note that $Pr(X_{(i_{j},i_{j+1})}= 1|\Pi_{k=0}^{j-1}X_{(i_{k},i_{k+1})}= 1 ) = 1$ if $(i_{j},i_{j+1}) = (i_{k},i_{k+1})$ or $(i_{j},i_{j+1}) =(i_{k+1},i_{k})$ for some $k \in [0..j-1]$ and  $Pr(X_{(i_{j},i_{j+1})}= 1|\Pi_{k=0}^{j-1}X_{(i_{k},i_{k+1})}= 1 ) = p_{i_{j}i_{j+1}}$ otherwise by independence

Now to prove the upperbound, we will modify the order in which we condition on edges in the path.  More specifically, from (\ref{eq:indicateE}) we have that 

\begin{multline}\label{eq:indicateE2}
 E[NBP_{r}(s,t)] = \sum_{\mathbf{i}\in\mathbf{B}}Pr(X_{(i_{0},i_{1})}=1)Pr(X_{(i_{r-1},i_{r})}=1|X_{(i_{0},i_{1})}=1)\LargerCdot\\  \Pi_{j=1}^{r-2}\Pr(X_{(i_{j},i_{j+1})} = 1|X_{(i_{r-1},i_{r})}\Pi_{k=0}^{j-1}X_{(i_{k},i_{k+1})}= 1),
\end{multline}

where we process the last edge immediately after the first edge and then resume the normal order for conditioning on the remaining edges in the path.  In particular it will be helpful to initially assume that $Pr(X_{(i_{r-1},i_{r})}=1|X_{(i_{0},i_{1})}=1) = p_{i_{r-1}i_{r}}$, that is $(i_{r-1},i_{r})$ is not an edge we have visited before.  Denote $NBP^{L}_{r}(s,t)$ as the number of nonbacktracking paths of length $r$ from node $s = i_{0}$ to node $t = i_{r}$, where the last edge does not equal the first edge.  We will argue that for $r \geq 2$,

\begin{equation}\label{eq:ineq}
NBP_{r}(s,t) \leq NBP_{r}^{L}(s,t) + NBP_{r-2}(s,t),
\end{equation}
where applying the (\ref{eq:ineq}) to itself iteratively yields the inequality,

\begin{equation}\label{eq:ineq2}
NBP_{r}(s,t) \leq NBP_{1}(s,t)+\sum_{m=2}^{r}NBP_{m}^{L}(s,t) ,
\end{equation}

Consequently, to derive a formula for the expected number of nonbacktracking paths of length $r$ from node $s$ to $t$, it suffices to construct a formula for the expected number of nonbacktracking paths of length $r$ from node $s$ to node $t$, where the first edge is not the same as the last edge.  (Note that computing the expected number of paths of length $1$ is precisely the probability that nodes $s$ and $t$ are neighbors).
\newline

To show (\ref{eq:ineq}), consider all (nonbacktracking) paths of length $r$ where the last edge is identical to the first edge.  It then follows that the path must be of the form $s,i_{1},....,i_{r-1},t$ as the first and last nodes in the path must be $s$ and $t$ respectively. Furthermore, since the first edge and last edge are identical and by assumption $s \neq t$, it follows that $i_{r-1} = s$ and $i_{1} = t$.  Hence all paths where the last and first edges are identical are of the form, $s,t,....,s,t$.  Assuming that an edge from node $s$ to $t$ exists, the number of such paths is precisely $NBP_{r-2}(t,s)$.  But since this is an undirected graph we have that $NBP_{r-2}(t,s) = NBP_{r-2}(s,t)$, which proves (\ref{eq:ineq}).

Since the first and last edges cannot be repeating edges, we can now invoke Lemma \ref{lem:probpath} to compute the probability that a given path exists.
Define $k_{0}$ to be the number of new edges.  For $i > 0$, let $k_{i}$ be the number of repeating edge blocks of length $i$ (where the first node has already been seen before).  So to compute $E[NBP_{r}^{L}(s,t)]$, we will fix (integer) values for $k_{i}$, consider all possible arrangements for each of the $k_{i}$ repeating edge blocks and then by accounting for all possible choices of nodes in the lists for the new edge interior $\mathbf{N}$ and in the list of nodes in a repeating edge block $\mathbf{R}$, an application of Lemmas \ref{lem:probpath} and \ref{lem:QSPcomp} yields that,

\begin{multline} \label{eq:E1}
E[NBP_{r}^{L}(s,t)] \leq \sum_{\substack{k_{0}+\sum_{i=1}^{r-2}i(k_{i})=r \\ \forall i\in[0..r-2], k_{i} \in [0..r]}}\binom{\sum_{i=0}^{r-2}k_{i}}{k_{0},k_{1},...,k_{r-2}} [\sum_{\mathbf{N},\mathbf{R}} \frac{d_{s}d_{t}}{S}\Pi_{i\in\mathbf{N}}\frac{d_{i}^{2}}{S}\Pi_{(j,k)\in\mathbf{R_1}}\frac{d_{j}d_{k}}{S}],
\end{multline}
where $\mathbf{R}_{1}$ consists of the nodes at the beginning and end of a repeating edge block is a function of $\mathbf{R}$ and the innersum represents all possible choices of nodes for constructing $\mathbf{N}$ and $\mathbf{R}$  that yield paths with the prescribed number of repeating edge blocks of various lengths.  We can then construct an upperbound to (\ref{eq:E1}) by identifying the nodes in $\mathbf{R_1}$ that must equal other nodes in the summation and bound the contrubition of that node's expected degree by $d_{max}$.  We claim that this yields the following inequality.
\begin{multline}\label{eq:E2}
E[NBP_{r}^{L}(s,t)] \leq \\ \sum_{k_{0}+\sum_{i}i(k_{i})=r}\binom{\sum_{i=0}^{r-2}k_{i}}{k_{0},k_{1},...,k_{r-2}}\sum_{\substack{j_1,...,j_{|\mathbf{N}|}}} \frac{d_{s}d_{t}}{S}\Pi_{l=1}^{|\mathbf{N}|}\frac{d_{j_{l}}^{2}}{S}\Pi_{m=1}^{r-2}(p_{max}(2r)^{m})^{k_{m}},
\end{multline}

where there are $r-2-\sum_{i}(i+1)(k_{i})$ nodes in $\mathbf{N}$ from Lemma \ref{lem:probpath}.  Note that the contribution from the term  $\Pi_{(j,k)\in\mathbf{R_1}}\frac{d_{j}d_{k}}{S}$ is replaced by $p_{max}$ as both nodes in $\mathbf{R_1}$ appear elsewhere by definition.  We can then account for the summation over all possible choices for the list $\mathbf{R}$, of nodes in a repeating edge block, by noting that for an arbitrary repeating edge block of length $l$, there are at most $(2r)^{l}$ choices for filling in the repeating edge block.  

Summing over all possible choices of nodes and using the fact that $\binom{\sum_{i=0}^{r-2}k_{i}}{k_{0},k_{1},...,k_{r-2}}\leq \frac{r^{\sum_{i\geq 1}k_{i}}}{\Pi_{i\geq 1}k_{i}!}$, yields the following upperbound for (\ref{eq:E2}).

\begin{multline} \label{eq:E1upperL}
E[NBP_{r}^{L}(s,t)]\leq \sum_{k_{0}+\sum_{i}i(k_{i})=r}r^{\sum_{i=1}k_{i}} \frac{d_{s}d_{t}}{S}\frac{S_2}{S}^{r-1-\sum_{i}(i+1)(k_{i})}\Pi_{l=1}^{r-2}\frac{(p_{max}(2r)^{l})^{k_{l}}}{k_{l}!} = \\
\sum_{k_{0}+\sum_{i}i(k_{i})=r} \frac{d_{s}d_{t}}{S}\frac{S_2}{S}^{r-1}\Pi_{l=1}^{r-2}\frac{(p_{max}(2r)^{l+1}\frac{S}{S_2}^{l+1})^{k_{l}}}{k_{l}!}
\end{multline}
We can then bound above (\ref{eq:E1upperL}) by removing the constraint under the summation by letting $k_{0} = r -\sum_{i=1}^{r-2}ik_{i}$ and allowing the other $k_{i}$ take on any non-negative integer value.
\begin{multline}\label{eq:E1upper2L}
E[NBP_{r}^{L}(s,t)]\leq \sum_{k_{1}=0,...,k_{r-2}=0}^{\infty} \frac{d_{s}d_{t}}{S}\frac{S_2}{S}^{r-1}\Pi_{l=1}^{r-2}\frac{(p_{max}(2r)^{l+1}\frac{S}{S_2}^{l+1})^{k_{l}}}{k_{l}!}= \\
\frac{d_{s}d_{t}}{S}\frac{S_2}{S}^{r-1}\Pi_{l=1}^{r-2}exp(p_{max}(2r)^{l+1}(\frac{S}{S_2})^{l+1})\leq \\
\frac{d_s d_t}{S}\frac{S_2}{S}^{r-1}exp(\frac{p_{max}(2r)^{2}\frac{S}{S_2}^{2}}{1-2r\frac{S}{S_2}}).
\end{multline}

Finally applying (\ref{eq:E1upper2L}) to (\ref{eq:ineq2}) yields the result. 
\end{proof}

From Lemma \ref{lem:exp} and Theorem \ref{thm:Epaths} we know that the expected number of simple or non-simple paths grows exponentially in terms of path length.   In particular, t for a flexible range of parameters in Chung-Lu random graph model, the diameter is no greater than $O(\log n)$,  \cite{chung2002average}. And since the number of paths grows exponentially (in terms of length), that for practical applicaiton, the length of the almost shortest paths will also be no greater than $O(\log n)$.
Consequently, we are interested in the ratio of the number of simple paths and non-simple paths, where the length $r = O(\log n)$.  To attain such results, we will need bounds on the variance for the number of simple and nonbacktracking paths; hence we have the following theorem.

\begin{thm}\label{thm:varSP}
Consider a collection of sources $\mathbf{S}$ and targets $\mathbf{T}$, where $\mathbf{S}\cap\mathbf{T} = \emptyset$ and denote $SP_{r}(\mathbf{S},\mathbf{T}) = \sum_{s\in\mathbf{S}}\sum_{t\in\mathbf{T}}SP_{r}(s,t)$.  (Define $NBP_{r}(\mathbf{S},\mathbf{T})$ analogously.)  Suppose that $r^{2}p_{max}\frac{S}{S_2} = o(1),$ $r\frac{S}{S_2} = o(1)$ and $\frac{1}{E(SP_{r}(\mathbf{S},\mathbf{T}))} = o(1)$.  Then
\begin{multline}
var(SP_{r}(\mathbf{S},\mathbf{T})) \leq E(SP_{r}(\mathbf{S},\mathbf{T}))^{2}[o(1) + (1+o(1))\frac{S}{S_2}(\frac{d_{max}}{\sum_{s\in\mathbf{S}}d_{s}}+\frac{d_{max}}{\sum_{t\in\mathbf{T}}d_{t}})].
\end{multline}
\end{thm}
\begin{proof}
We provide a sketch of the proof.  Intuitively, $SP_{r}(\mathbf{S},\mathbf{T}) = \sum X_{\alpha}$ is a sum of Bernoulli random variables (corresponding to the existence of a simple path) each with a low probability of success.  Suppose we only consider the pairs of Bernoulli random variables that are independent that contribute to the variance.  We could then directly approximate the summation as a Poisson random variable.  But since the variance of a Poisson random variable equals the square of the expected value and we assumed that $\frac{1}{E(SP_{r}(\mathbf{S},\mathbf{T}))} = o(1)$, this contribution relative to the expected value squared, is negligible.  Hence, we are only interested in identifying the dependent pairs of Bernoulli random variables $X_{\alpha_1},X_{\alpha_2}$, from the summation  $\sum X_{\alpha}$.  

\noindent \textbf{Case 1:} If $X_{\alpha_1}, X_{\alpha_2}$ share a common edge and that edge is not the first or last edge of $X_{\alpha_2}$, we claim that the contribution to the variance, is negligible relative to .$E(SP_{r}(\mathbf{S},\mathbf{T}))^{2}$. Fix an $\alpha_{1}$.  Then define the set $D(\alpha_{1})$ to consist of all indices $\beta\neq\alpha_1$ such that $X_{\alpha_1},X_{\beta}$ are dependent random variables and fall under Case 1.  We will justify that 
\begin{equation}
E(\sum_{\alpha_1}X_{\alpha_1}\sum_{\beta\in D(\alpha_{1})}X_{\beta}) \approx \sum_{\alpha_1}E(X_{\alpha_1})[E(NBP_{r}(\mathbf{S},\mathbf{T}))-E(SP_{r}(\mathbf{S},\mathbf{T}))] ,
\end{equation}
as if we only consider $X_{\beta}$ such that $\beta \in D(\alpha_{1})$, we are stipulating that one of the nodes repeats in the path.  This contribution should be similar to the difference of the number of nonbacktracking paths and simple paths, as the difference identifies all paths where nodes may appear more than once.  Since this difference between the number of nonbacktracking paths and simple paths is at most $o(1)E(SP_{r}(\mathbf{S},\mathbf{T}))$,, we conclude that 

\begin{equation}
E(\sum_{\alpha_1}X_{\alpha_1})o(1)E(SP_{r}(\mathbf{S},\mathbf{T})) = E(SP_{r}(\mathbf{S},\mathbf{T}))^{2}o(1)
\end{equation}
\noindent \textbf{Case 2:} Suppose that $X_{\alpha_1}, X_{\alpha_2}$ share a common edge and that edge is the first or last edge of $X_{\alpha_2}$.  (Note that from Theorem \ref{thm:Epaths} when we computed $E(NBP_{r}(\mathbf{S},\mathbf{T}))$, we managed to circumvent this issue by assuming the first and last edge do not repeat.)  Without loss of generality, we will assume that the first edge in $X_{\alpha_2}$ repeats.  In the worst case scenario, the expected degree of the second node in the path is $d_{max}$.  Define $\Delta(\alpha_{1})$ to consist of all indices $\beta\neq\alpha_1$ such that $X_{\alpha_1},X_{\beta}$ are dependent random variables and fall under this subcase of Case 2 (where the first edge repeats).  Let node $s_{max}$ have expected degree $d_{max}$.  Consequently, we claim that

\begin{equation}E(\sum_{\alpha_1}X_{\alpha_1}\sum_{\beta\in \Delta(\alpha_{1})}X_{\beta}) \approx E(\sum_{\alpha_1}X_{\alpha_1})\sum_{t\in\mathbf{T}}E(SP_{r-1}(s_{max},t)) = 
\end{equation}
\begin{equation}
\sum_{s\in\mathbf{S},t\in\mathbf{T}}E(SP_{r}(\mathbf{S},\mathbf{T}))\sum_{t\in\mathbf{T}}E(SP_{r-1}(s_{max},t)),
\end{equation}

where the above equation follows as the inclusion of additional repeating edges is negligible (due to the argument from \textbf{Case 1}) and by assuming the worst case scenario that the second node in the simple path corresponding to $X_{\alpha_{1}}$ has expected degree $d_{max}$, we compute the number of simple paths of length $r-1$ from $s_{max}$ to $t$ for each $t\in\mathbf{T}$.  By an application of Theorem \ref{thm:Epaths} and Lemma \ref{lem:exp}, this gives us a contribution of  $$E(SP_{r}(\mathbf{S},\mathbf{T}))^{2}(1+o(1))\frac{S}{S_2}(\frac{d_{max}}{\sum_{s\in\mathbf{S}}d_{s}})$$ to the variance.
\newline\newline
Considering the case where the last edge has already been visited before is analogous.
\end{proof}

At this juncture, we can show that if the source and target have sufficiently large expected degree, then the number of simple paths of prescribed length connecting the two nodes should approximate the number of nonbacktracking paths.  
\begin{cor}\label{cor:main1}
Consider two nodes $s,t$ with expected degrees $d_{s},d_{t}$.  Suppose that $\frac{d_{max}S}{S_2\min(d_{s},d_{t})} = o(1)$, $r\frac{S}{S_2} = o(1)$, $r^{2}\frac{S}{S_2}p_{max}=o(1)$ and $\frac{1}{E(SP_{r}(s,t))} = o(1)$; then with high probability,

$$\frac{NBP_{r}(s,t)}{SP_{r}(s,t)} = 1 + o(1).$$
\end{cor}
\begin{proof}
	To prove the statement, it suffices to show that \begin{equation}\label{eq:keycor} \frac{NBP_{r}(s,t)-SP_{r}(s,t)}{SP_{r}(s,t)}\rightarrow 0.
	\end{equation}
	
	First we invoke Chebyshev's Inequality and Theorem \ref{thm:Epaths} to show that with high probability 
	$SP_{r}(s,t) \rightarrow E(SP_{r}(s,t))$ and substitute the denominator in (\ref{eq:keycor}) with its expected value.  By invoking Theorem \ref{thm:Epaths} and Lemma \ref{lem:exp} we can show that with high probability $NBP_{r}(s,t)-SP_{r}(s,t)$ is much smaller than $E(SP_{r}(s,t))$ using Markov's inequality.   
\end{proof}

We now wish to extend Corollary \ref{cor:main1}, where $d_{s}$ (or $d_{t}$) may not satisfy the condition that $\frac{d_{max}S}{S_2d_{s}} = o(1).$  To construct such an extension, define the $kth$ neighborhood of the node $s$, $N_{k}(s) = \{ i \in V: d(s,i) = k\}$.  We first note that if $s$ is in the giant connected component, there exists a $k$, such that the expected number of edges from the $kth$ neighborhood of $s$, $\sum_{i \in N_{k}(s)}d_{i}$ is sufficiently large.  Consequently, if we apply Corollary \ref{cor:main1} to the $kth$ neighborhood of $s$, we can bound the ratio of the number of simple and nonbacktracking paths.  This leads us to our main result.  

\begin{cor}\label{cor:main}
Suppose $s$ and $t$ are part of the giant connected component.  If the following conditions hold: 
\begin{itemize}
\item $r\frac{S}{S_2} = o(1)$
\item $r^{2}\frac{S}{S_2}p_{max}=o(1)$
\item $\frac{1}{E(SP_{r}(s,t))} = o(1)$
\item $\frac{d_{max}\ln\ln N}{\sqrt{S}}= o(1)$,
\end{itemize}
then with high probability, $$\frac{NBP_{r}(s,t)}{SP_{r}(s,t)} = 1 + o(1).$$
\end{cor}
\begin{proof}
First we will prove that given that $s$ is part of the giant connected component, then with high probability there exists a $k$ such that the subgraph formed by the nodes $N_{\leq k}(s) = \cup_{j=0}^{k}N_{j}(s)$ is a tree and the expected number of edges in the subgraph $N_{\leq k}(s)$ is bounded between $\frac{d_{max}S}{S_2}\sqrt{\ln\ln N}$ and $d_{max}\ln\ln N$. 

To see this, since $s$ is part of the giant component, there exists a first $k$ such that $\sum_{i \in N_{\leq k-1}(s)}d_{i}\leq \frac{d_{max}S}{S_2}\sqrt{\ln\ln N}\leq \sum_{i \in N_{\leq k}(s)}d_{i}$.  An application of Markov's Inequality and Theorem \ref{thm:Epaths} shows that by using the fact that $\sum_{i \in N_{\leq k-1}(s)}d_{i}\leq \frac{d_{max}S}{S_2}\sqrt{\ln\ln N}$, with high probability  $$\sum_{i \in N_{\leq k}(s)}d_{i}\leq d_{max}\ln\ln N.$$

Furthermore, since $\sum_{i \in N_{\leq k}(s)}d_{i}\leq d_{max}\ln\ln N$ and by assumption $\frac{d_{max}\ln\ln N}{\sqrt{S}}= o(1)$, with high probability, the subgraph formed by $N_{\leq k}(s)$ must be a tree.  

Applying this fact both to nodes $s$ and $t$, tells us that there exists a $k$ and $l$ such that $N_{\leq k}(s)$ and $N_{\leq l}(t)$ are trees and the expected number of edges is bounded between $\frac{d_{max}S}{S_2}\sqrt{\ln\ln N}$ and $d_{max}\ln\ln N$.  Furthermore, with high probability $N_{\leq k}(s)\cap N_{\leq l}(t) = \emptyset$.  We then have that

\begin{multline}\label{eq:asymp}
\sum_{\substack{s_{*}\in N_{k}(s)\\ t_{*}\in N_{l}(t)}}SP_{r-k-l}(s_*,t_*) \leq SP_{r}(s,t) \leq NBP_{r}(s,t) \leq \\ \sum_{\substack{s_{*}\in N_{k}(s)\\ t_{*}\in N_{l}(t)}}NBP_{k}(s,s_*)NBP_{r-k-l}(s_*,t_*)NBP_{l}(t_*,t).
\end{multline}

But since the subgraphs formed by $N_{\leq k}(s)$ and $N_{\leq l}(t)$ are trees (with high probability), for every $s_* \in N_{k}(s)$ and for every $t_* \in N_{l}(t)$, 
\begin{equation} \label{eq:tree} 1 = NBP_{k}(s,s_*) = NBP_{l}(t_*,t).
\end{equation}
Substituting (\ref{eq:tree}) into (\ref{eq:asymp}) tells us that with high probability,

\begin{equation}\label{eq:asymp2}
\sum_{\substack{s_{*}\in N_{k}(s)\\ t_{*}\in N_{l}(t)}}SP_{r-k-l}(s_*,t_*) \leq SP_{r}(s,t) \leq NBP_{r}(s,t) \leq \sum_{\substack{s_{*}\in N_{k}(s)\\ t_{*}\in N_{l}(t)}}NBP_{r-k-l}(s_*,t_*).
\end{equation}

Finally applying Theorem \ref{thm:varSP} tells us that, 

\begin{equation}\label{eq:asymp3}
\sum_{\substack{s_{*}\in N_{k}(s)\\ t_{*}\in N_{l}(t)}}SP_{r-k-l}(s_*,t_*) \leq SP_{r}(s,t) \leq NBP_{r}(s,t) \leq (1+o(1))\sum_{\substack{s_{*}\in N_{k}(s)\\ t_{*}\in N_{l}(t)}}SP_{r-k-l}(s_*,t_*).
\end{equation}

Consequently, we conclude that with high probability, $\frac{NBP(s,t)}{SP(s,t)} = 1 + o(1).$

\end{proof}
Now that we have shown that the number of nonbacktracking paths (of prescribed length) approximates the number of simple paths for a wide range of parameters under the Chung-Lu random graph model, we illustrate how to compute almost shortest nonbacktracking paths using the algorithm from Section 2.  In particular, we computed {\em paths} bounded by length $D$ from node $s$ to $t$ by considering a partial path $(x_{m},x_{m-1}....,x_{1},x_{0})$, where $x_{0} = t$, measuring the distance traveled so far $D_{*} = \sum_{i=0}^{m-1}d(x_{m-i},x_{m-i-1})$ and adding a new node to the partial path $x_{m+1}$ if $x_{m+1}$ is a neighbor of $x_{m}$ and $d(s,x_{m+1})+d(x_{m+1},x_{m}) \leq D - D_{*}$.  Iteratively adding nodes in this way would yield a path from node $s$ to node $t$ with length at most $D$.  Consequently to find sufficiently short nonbacktracking paths, it will be helpful to compute the minimum length of a nonbacktracking path between two nodes under some constraints.  This motivates the following definition.

\begin{defn} For any edge $(a,b)$ in the graph $G$, denote $d_{NBP}(a\rightarrow b,t)$ as the length of the shortest nonbacktracking path of the form $(a,b,...,t)$. 
\end{defn}

Now to efficiently compute $d_{NBP}(a\rightarrow b,t)$, it will be helpful to consider a (directed) shortest path tree $T = (V,E_{T})$ for the graph $G$. [As $T$ is directed, note that if the edge $(a,b)$ appears in the graph $T$, this does not imply that $(b,a)$ appears in $T$.  Furthermore, given an edge $(a,b)$ in a directed graph, we read the edge as going from node $a$ to node $b$; that is, $a$ is an incoming neighbor of $b$.]   In particular, if $(b,a)$ is an edge in $T$, then the shortest path of the form $(a,b,...,t)$ is $(a,b,a,...,t)$, which would not be nonbacktracking.  Consequently, we claim that we have the following recursion for computing $d_{NBP}(a\rightarrow b,t)$ when $(b,a)$ is an edge that appears in $T$.  

\begin{lem}\label{lem:QSPcomp}For a positively weighted graph $G$,

\[ d_{NBP}(a\rightarrow b,t) = d(a,b) + \min_{n\in Nbr(b)}  \left\{
\begin{array}{ll}
       d(b,n) + d(n,t) & n \neq a \hspace{3pt} \text{and}\hspace{3pt} (n, b) \notin E_{T}  \\
      d_{NBP}(b\rightarrow n,t) & n \neq a \hspace{3pt}\text{and}\hspace{3pt} (n, b) \in E_{T} \\
      \infty & n = a\\
\end{array} 
\right. \]
\end{lem}
\begin{proof}
To find the length of the shortest nonbacktracking path from $a$ to $t$ where the first edge is $(a,b)$, we can look at the lengths of the shortest nonbacktracking path of the form $(a,b,n,...,t)$, where $n$ is a neighbor of $b$ and we minimize over all choices for $n$. \newline\newline 
 If $n = a$, then there is no such nonbacktracking path of the form $(a,b,n...,t)$ and we define the length as $\infty$.
\newline \newline \noindent
Alternatively, if $n\neq a$ and $(n,b) \notin E_{T}$, then it follows that the shortest path of the form $(b,n,...,t)$ is a simple path (and hence nonbacktracking).  In particular the length of the path is $d(b,n) + d(n,t)$.
\newline \newline \noindent
Finally, suppose that $(n,b) \in E_{T}$, $n\neq a$ and consider a nonbacktracking path of the form $(a,b,n,...,t)$.  Under these conditions, it follows that $(a,b,n,...,t)$ is nonbacktracking if and only if the path formed by deleting the first two edges of $(a,b,n,...,t)$, $(n,...,t)$, is a nonbacktracking path as well.  Consequently we can denote the length of the shortest nonbacktracking path of the form $(a,b,n,...,t)$ as 
$$d(a,b) + d_{NBP}(b\rightarrow n,t).$$ 
\end{proof}

\noindent \textbf{Remark:} In the event that $(b,a)$ is not an edge in $E_{T}$, then the shortest nonbacktracking path of the form $(a,b,...,t)$ is a simple path and hence $d(a\rightarrow b,t) = d(a,b) + d(b,t)$.  Otherwise, (as mentioned previously) if $(b,a)$ is an edge in $E_{T}$, then Lemma \ref{lem:QSPcomp} is especially helpful for computing, $d_{NBP}(a\rightarrow b,t)$. To efficiently solve for the length of the shortest nonbacktracking path of the form $(a,b,...,t)$, we start by using Lemma \ref{lem:QSPcomp} to solve for $d_{NBP}(a\rightarrow b,t)$, where $b$ has $0$ incoming edges in $T$. We can subsequently solve for the remaining nonbacktracking path distances by identifying nodes $b$ such that for all incoming neighbors $n$ in $T$, $d_{NBP}(b\rightarrow n,t)$ is known and invoke Lemma \ref{lem:QSPcomp} to compute the nonbacktracking path distance.  Continuing this process yields an $O(m+n)$ algorithm for computing the lengths of the shortest nonbacktracking paths between two nodes, where the first edge is fixed.  

We can then generalize \hyperlink{alg1}{Algorithm 1} to find only nonbacktracking paths in the following manner.  Intially, we compute the lengths of the shortest nonbacktracking paths $d_{NBP}(a\rightarrow b,t)$, for all edges $(a,b)$.  Subsequently, for each node $a$, we sort $a$'s neighbors,$n$ according to $d_{NBP}(a\rightarrow n,t)$.  Then once we have determined that there exists a nonbacktracking path of length bounded by $D$ of the form $(x_{0},....,x_{m},...,t)$, we can determine if there exists a nonbacktracking path of length bounded by $D$ of the form $(x_{0},....,x_{m},x_{m+1},...,t)$, where $x_{m+1}$ is a neighbor of $x_{m}$ by checking that $x_{m+1} \neq x_{m-1}$ and that $d_{NBP}(x_{m}\rightarrow x_{m+1},t) \leq D - \sum_{i=0}^{m-1}d(x_{i},x_{i+1})$.  It then follows from the analysis of \hyperlink{alg1}{Algorithm 1}, that the time complexity for identifying almost shortest nonbacktracking paths between two nodes is the same for finding almost shortest paths between two nodes.  

\subsection{Simulations for the Ratio of Paths to Simple Paths}

For an undirected graph, the presence of nodes of high degree can influence the ratio of the number of simple paths to non-simple paths, as illustrated in \cite{castellano2010thresholds,castellano2017topological}.  For this reason, we considered the problem of computing almost shortest nonbacktracking  paths, as they are easy to compute and for a more flexible parameter regime, the number of simple paths asymptotically approximates the number of nonbacktracking paths of the same length in the Chung-Lu random graph model.  

Even so, for many graphs we can approximate the number of simple paths by the number of paths between two nodes. To better understand the relationship,  we ran numeric simulations.  Intuitively speaking, given a collection of graphs with a fixed number of high degree nodes (to ensure a substantial number of nonsimple paths for sufficiently large $r$), the claim is that graphs associated with a larger $\frac{S_2}{S}$  will result in a larger percentage of simple paths of length $r$; that is, $\frac{S_2}{S}$ is related to the expected number of neighbors of a node and increasing the number of neighbors of a node will result in more new (simple) paths.  

To justify this claim, we constructed realizations of Chung-Lu graphs with expected degree sequences where we fixed the average degree, varied $\frac{S_2}{S}$, and selected a fixed number of nodes to have expected degree equal to $\sqrt{S}$ .  In particular, we used a Markov Chain Monte Carlo method similar to \cite{lu2011generating} for randomly generating expected degree sequences such that any realization has a fixed expected average degree of $8$. Expected degree sequences satisfied a specified value for $\frac{S_2}{S}$ and consisted of at least four nodes with an expected degree of  $\sqrt{S}=\sqrt{8n}$, where $n$ is the number of nodes.  In these simulations, graphs consisted of $800$ nodes. We constructed $100$ such expected degree sequences.  Subsequently, from each expected degree sequence we constructed a realization from the Chung-Lu random graph model.  We then randomly chose 50 pairs of nodes (each with a minimum degree of 5) and calculated the ratio of the number of simple and non-simple almost shortest paths for various lengths.  Figure \ref{fig:ratio} presents clusters of three box plots of the ratio corresponding to realizations from each of these expected degree sequences.

\begin{figure}[h]
\centering
\includegraphics[scale=.5]{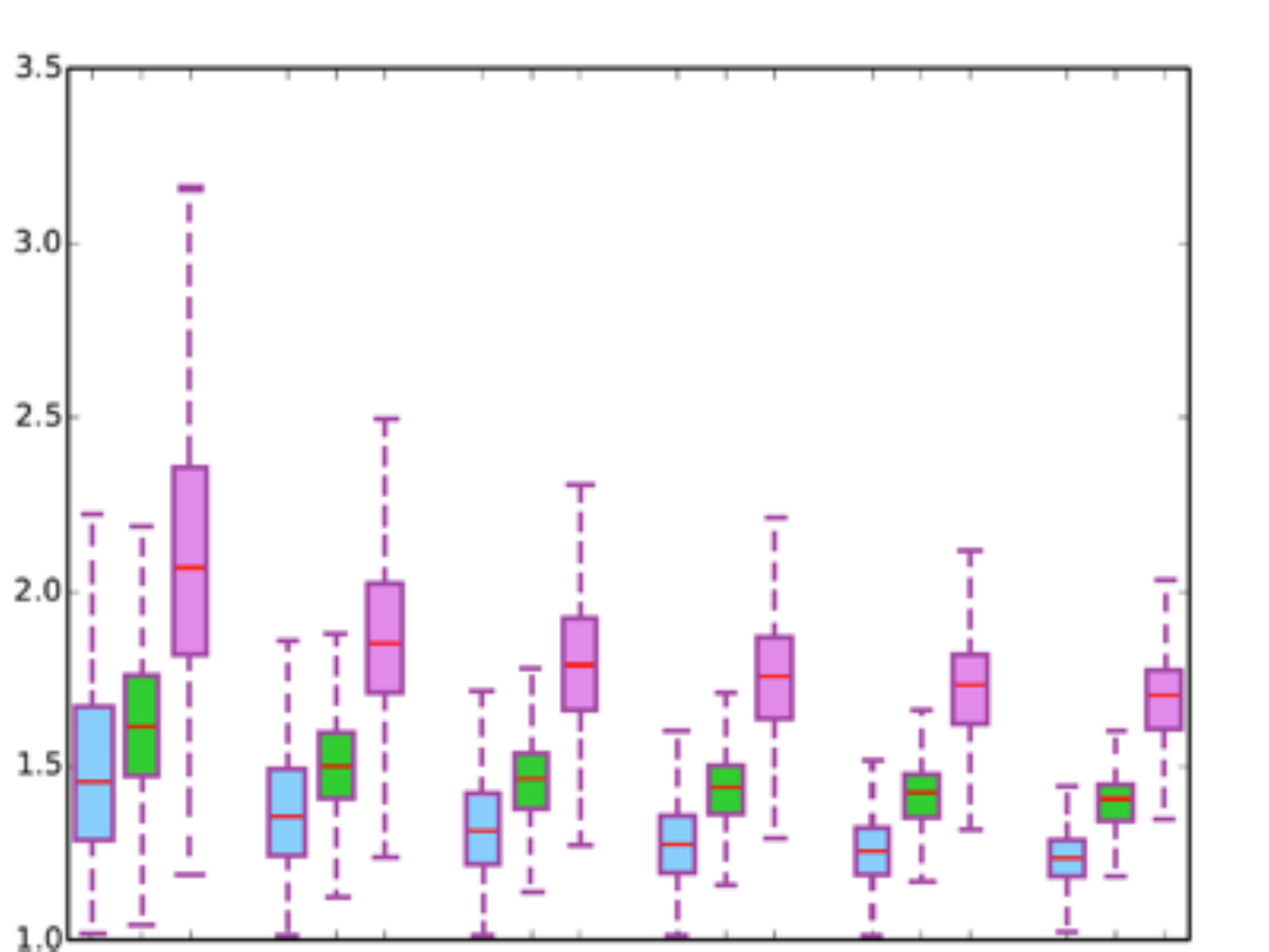}
\caption{A box plot of the ratio of the number of paths of length $r$ compared to the number of simple paths of length $r$.  Each cluster of three box plots represents different values of $r$, minimum distance $+ 2$, $+3$ and $+4$.  Similarly a more rightwards cluster of box plots on the x-axis denotes an increased value of $\frac{S_2}{S}$ in the expected degree sequence. Note that there appears to be an exponential growth in the number of non-simple paths (relative to the number of paths) as $r$ increases but that for modest values of $r$, this quantity is well controlled.} \label{fig:ratio}
\end{figure}

From the simulations, it appears that it is far more efficient to calculate the number of simple paths of length $r$ by calculating the number of  paths of length $r$ (and then removing paths that aren't simple) than using one of the existing algorithms for computing the number of simple paths directly as referenced in the introduction.  Furthermore, while the ratio of non-simple paths to simple paths grows as we increase $r$, in practice we must compute  {\em exponentially many} paths to see an exponential growth in the penalty for computing both simple and nonsimple paths.    

\section{Connectivity Simulations in the AS Graph}
We now consider an application of the almost shortest (simple) path problem to internet routing.  More precisely, we wish to inquire the robustness of the Autonomous System (AS) Graph and some random graph models to an edge deletion process and assess the connectivity of the graph.

Measuring {\em{connectivity}} for this application is rather ambiguous. For example under an edge deletion process \cite{lopez2007limited}  construes connectivity between two nodes as a measure that solely depends on the existence of a path of length bounded by some constant multiple of the diameter.  Alternatively, one can consider measuring connectivity by requiring the existence of a giant component or dynamical robustness   \cite{motter2002cascade,crucitti2004error,duch2007effect,he2009effect}.  

\begin{figure}[!htb]
\centering
\begin{subfigure}[h]{.4\textwidth}
\centering
\includegraphics[width=\linewidth]{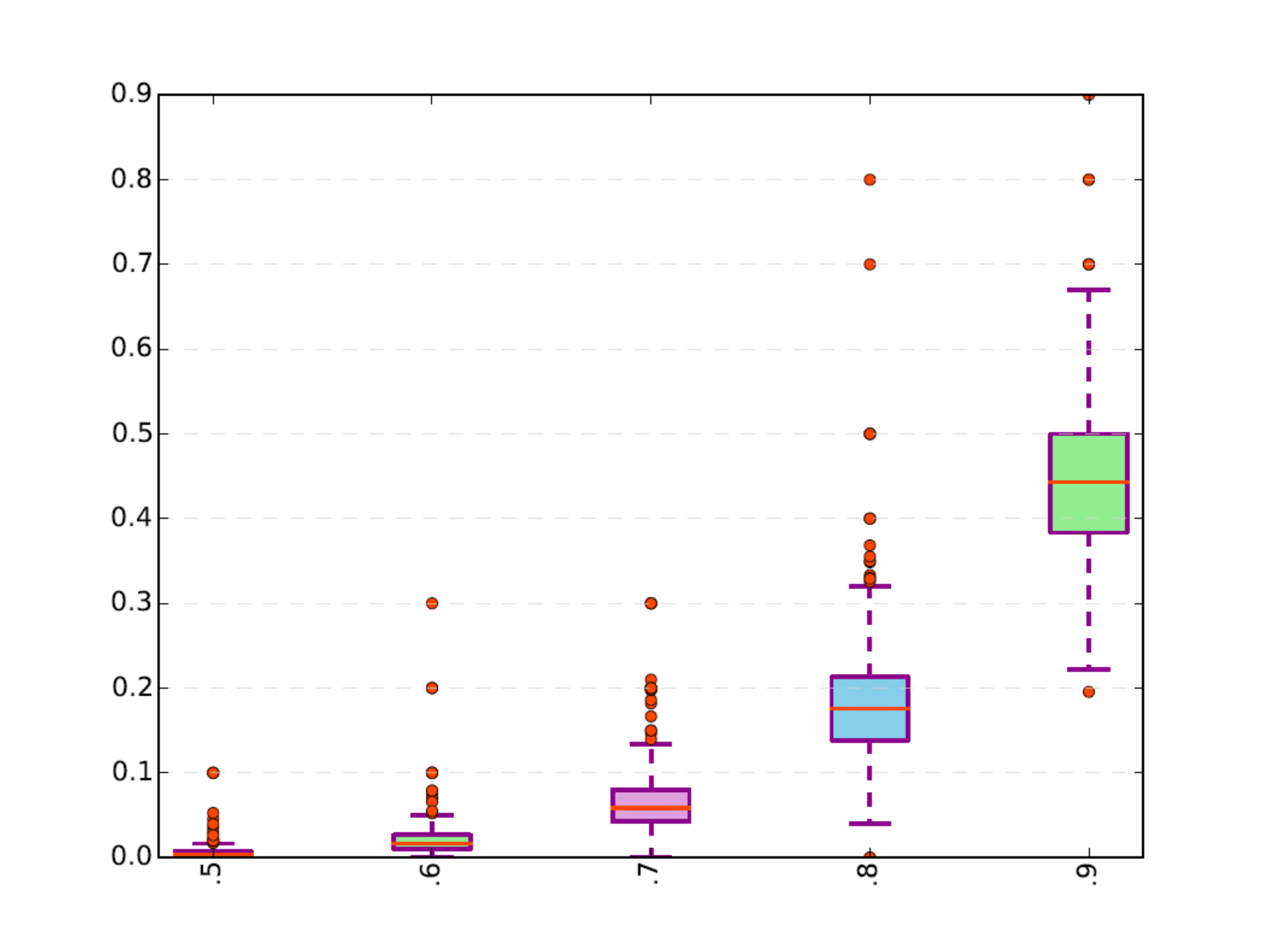}
        \caption{Box Plot of Percentage of Functional Number Almost Shortest Paths for an Erdos-Renyi Random Graph.}\label{fig:ERBox}
\end{subfigure}
\begin{subfigure}[h]{.4\textwidth}
\centering
\includegraphics[width=\linewidth]{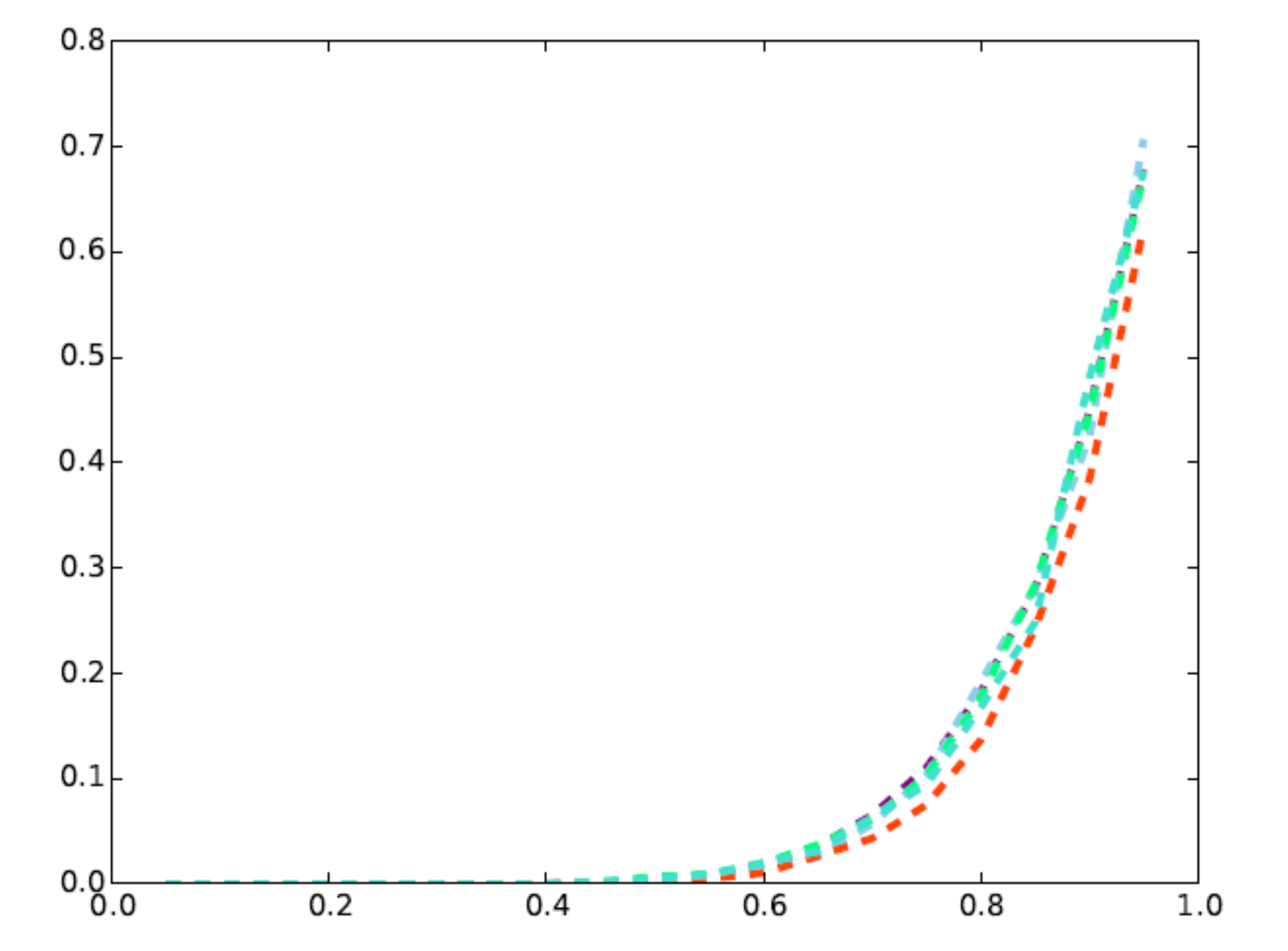}
	\caption{Median Percentage of Functional Almost Shortest Paths for an Erdos-Renyi Random Graph for 5 Different Nodes}\label{fig:ERMed}
\end{subfigure}

\begin{subfigure}[h]{.4\textwidth}
\centering
\includegraphics[width=\linewidth]{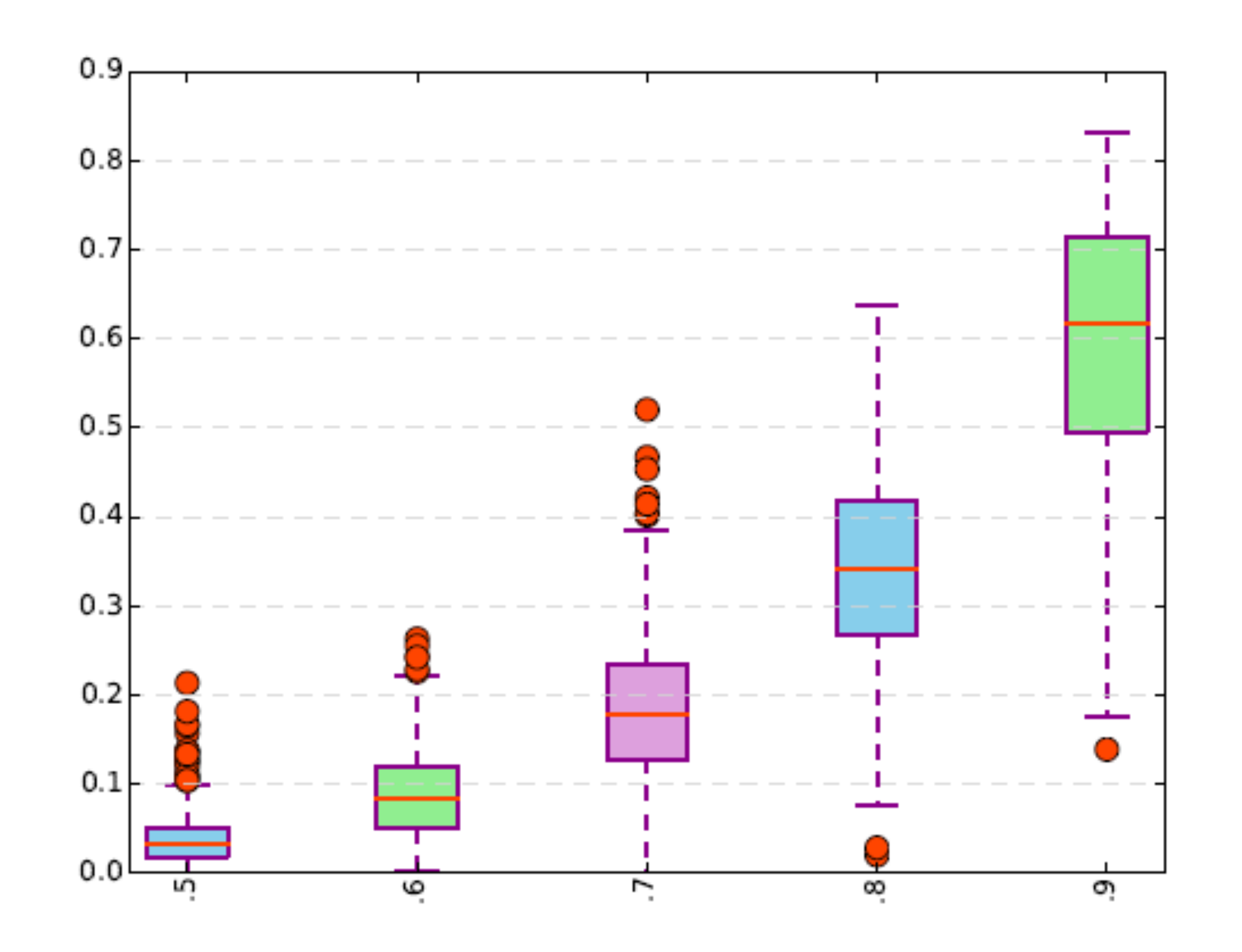}
	\caption{Box Plot of Percentage of Functional Number Almost Shortest Paths for a Chung-Lu Random Graph.}\label{fig:CLuBox}
\end{subfigure}
\begin{subfigure}[h]{.4\textwidth}
\centering
\includegraphics[width=\linewidth]{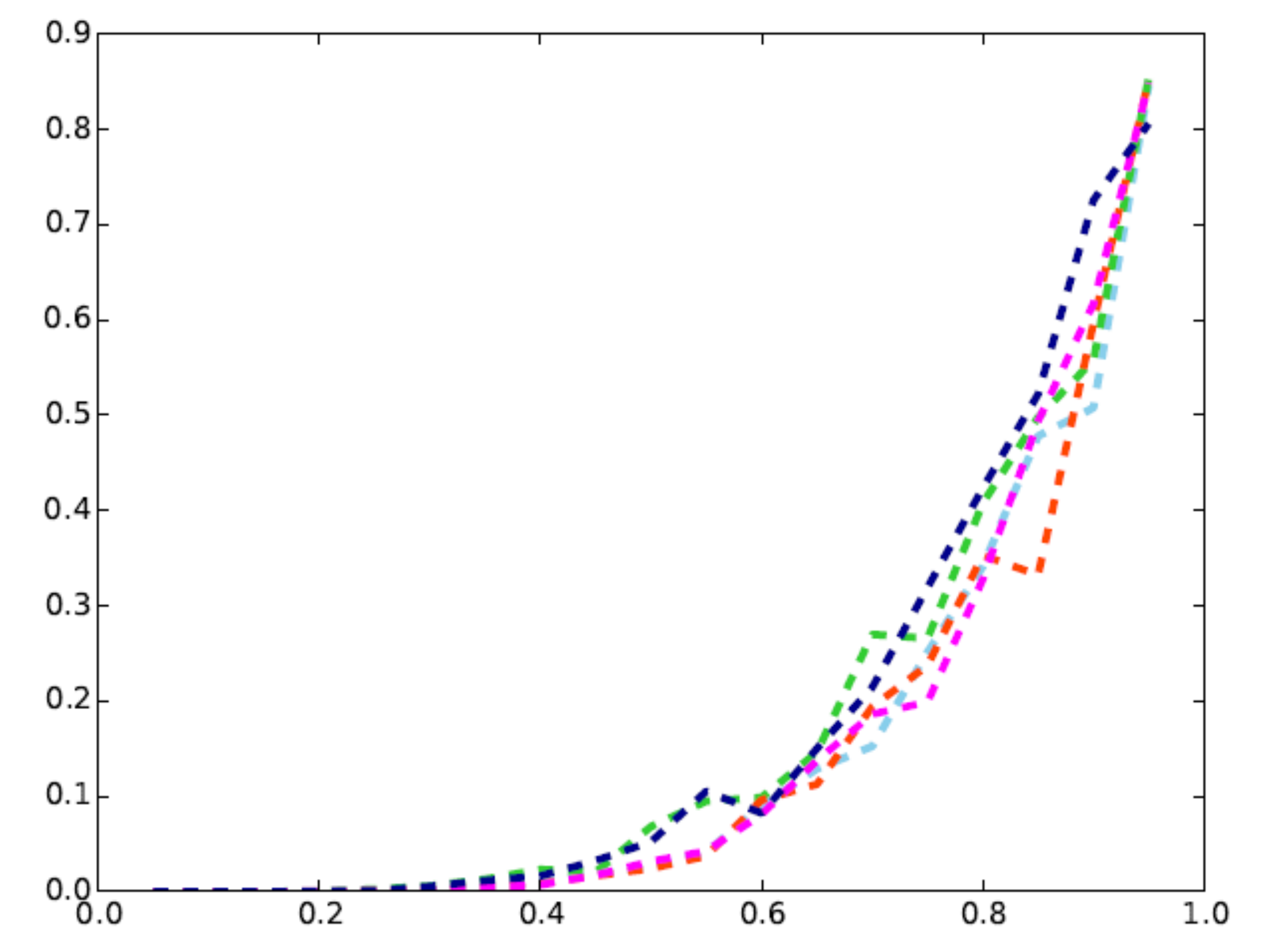}
\caption{Median Percentage of Functional Almost Shortest Paths for a Chung-Lu Random Graph for 5 Different Nodes}\label{fig:CLuMed}
\end{subfigure}

\begin{subfigure}[h]{.4\textwidth}
\centering
\includegraphics[width=\linewidth]{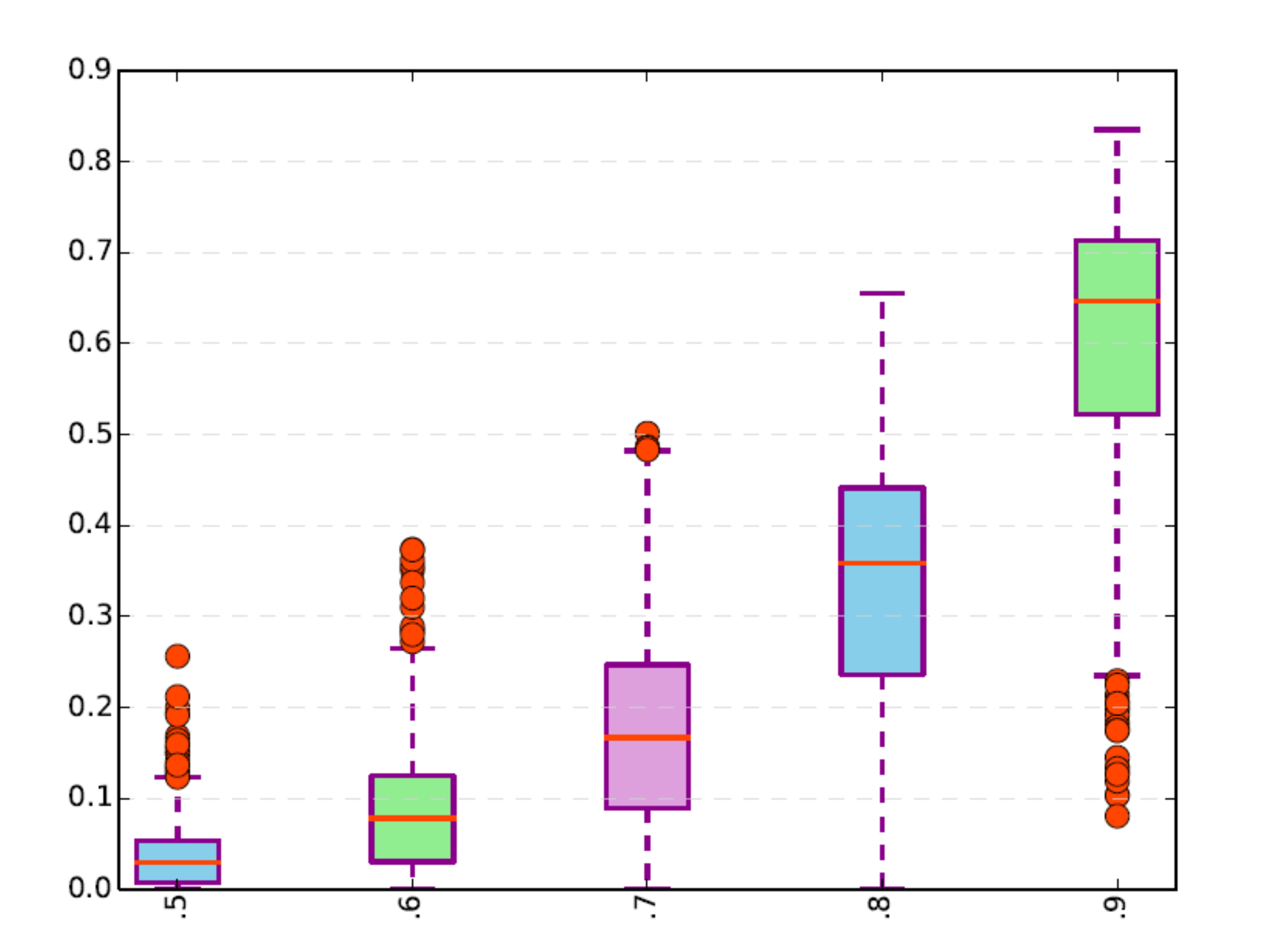}
	\caption{Box Plot of Percentage of Functional Number Almost Shortest Paths for the AS Graph}\label{fig:ASBox}
\end{subfigure}
\begin{subfigure}[h]{.4\textwidth}
\centering
\includegraphics[width=\linewidth]{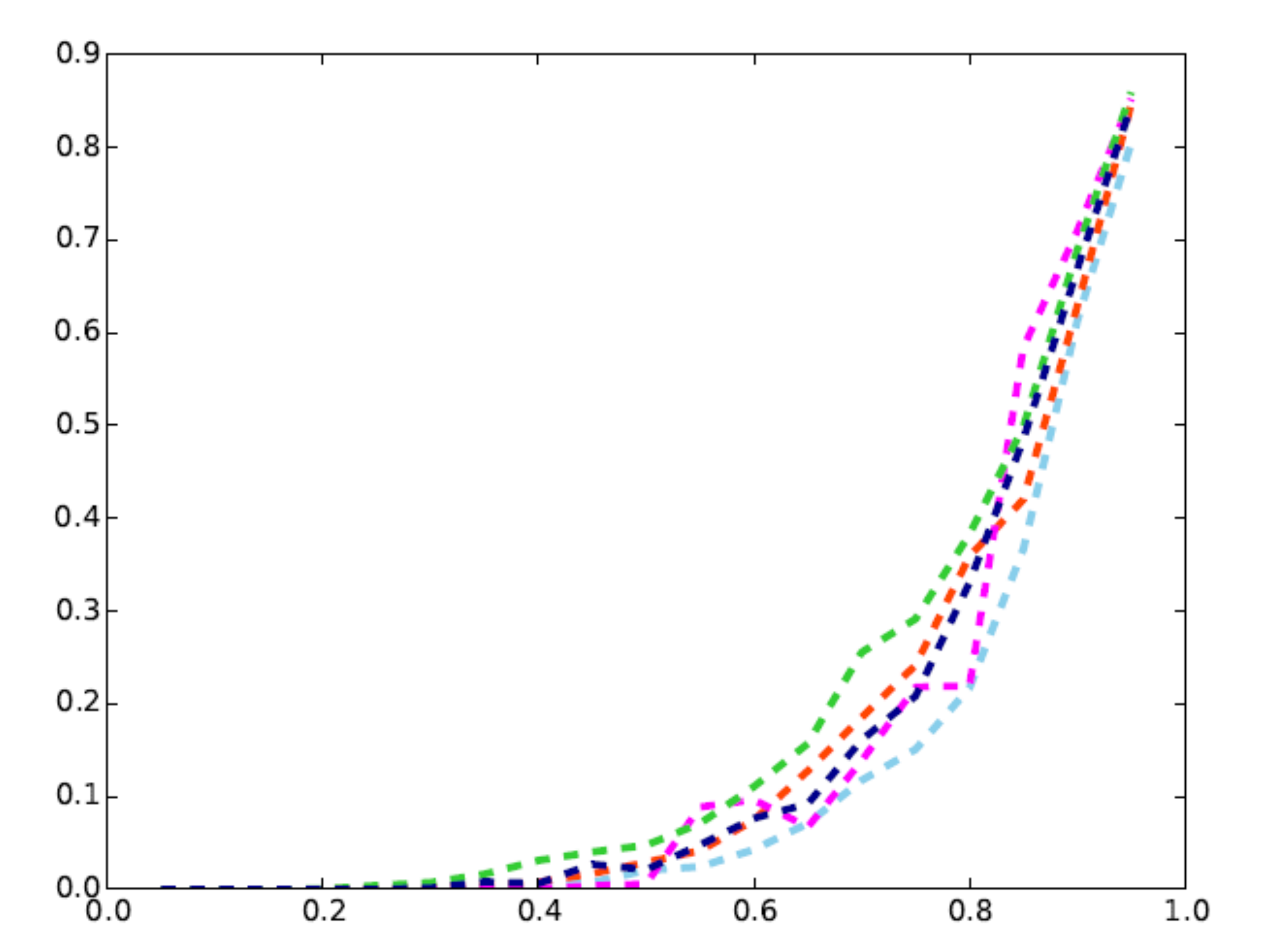}
\caption{Median Percentage of Functional Almost Shortest Paths for the AS Graph for 5 Different Nodes}\label{fig:ASMed}
\end{subfigure}
\caption{}
\label{fig:SIMS}
\end{figure}

In this work as dynamical robustness (and path existence) may fail to capture the potential ramifications of the existence of only a modest number of short, viable paths, we would instead like to track the percentage (or number) of surviving almost shortest paths under an edge deletion process, where we delete each edge from the graph independently with probability $p$.  On the left column of Figure \ref{fig:SIMS},\ref{fig:ERBox}, \ref{fig:CLuBox} and \ref{fig:ASBox}, we plotted box plots for the percentage of surviving almost shortest paths (y-axis) under an edge deletion process with probability $p$ denoted on the x-axis.  More specifically, we sampled 20 random pairs of nodes, one with 10 edges and another with 12 edges and repeated the edge deletion process 20 times for each of the 20 distinct node pairs with a given $p$.  When considering a collection of almost shortest paths, in practice we only included paths that were at most 3 or 4 edges longer than the path of minimal length.  For a particular pair of nodes, a collection of almost shortest paths could consist of more than 50 million paths.
In figure \ref{fig:ERBox}, we construct an Erdos-Renyi random graph with average degree chosen to match the AS Graph.  Subsequently, in figure \ref{fig:CLuBox}, we consider a Chung-Lu random graph with an expected degree sequence chosen to match the AS Graph as well \cite{miller2011efficient,hagberg2015fast}.  Finally in figure \ref{fig:ASBox}, we consider a snapshot of the AS Graph from January 2015, where we compiled edges based on route announcements from the Ripe and Route Views data set.

On the right column of Figure \ref{fig:SIMS}, we consider five randomly chosen node pairs and plot the median percentage of surviving paths.  Two perhaps surprising results emerge from Figure \ref{fig:SIMS}.  Firstly, for a given $p$, the median percentage of surviving paths appear to be roughly the same across all pairs of nodes (of sufficiently high degree) in spite of the fact that the existence of two paths under an edge deletion process is often dependent on one another  And secondly, while the Erdos-Renyi graph fails to capture the distribution of the percentages of surviving paths of almost shortest length in the AS Graph, the Chung-Lu random graph model behaves remarkably similar to the AS Graph and heavily suggests that knowledge of the degree sequence plays a fundamental role in predicting the percentage of surviving almost shortest paths.    

\section{Conclusions}
Identifying almost shortest paths between two nodes arises in numerous applications including internet routing and epidemiology.  Since we want to find {\em many} almost shortest paths in these real world networks, we would like our algorithm to exploit properties commonly found in these networks.  Consequently, we provided a simple algorithm for computing all  paths bounded by length $D$ between two nodes in an graph with $m$ weighted edges..  In particular, we demonstrated that the space and time complexity is $O(m\log m+kL)$, where $L$ is an upperbound for the number of nodes that appear in any almost shortest path, for graphs that exhibit certain real world network features.  

For many applications, we instead want to find the almost shortest {\em simple} paths, where we cannot visit a node more than once in a path.  Since computing almost shortest simple paths can be computationally expensive, we presented a rigorous framework for explaining when we could use a variant of our solution to solve the almost shortest simple paths problem.  More specifically, we analyzed the Chung-Lu random graph model, which emulates many of the properties frequently observed in real world networks, and demonstrated in Corollary \ref{cor:main} that for a flexible choice of parameters, we can approximate the number of simple paths between two nodes with the number of nonbacktracking paths of the same length.  We demonstrated how to modify the algorithm to efficiently compute almost shortest nonbacktracking paths.  And subsequently, we performed numeric simulations illustrating that the ratio of the number of paths to simple paths is well behaved in Chung-Lu random graphs.

In an effort to provide rigorus arguments supporting the efficiency of our algorithm for solving the almost shortest simple paths problem on the Chung-Lu random graph model, other questions organically emerged in the process.  While in this work we focused primarily on properties of the number of simple paths to nonbacktracking paths for Chung-Lu random graphs, we could ask about the ratio of simple paths to paths for other random graph models as well.  

Finally, we considered an application to internet routing where we would like to assess the quality of the connectivity between two nodes under an edge deletion process.  To measure the quality of connectivity we constructed large collections of almost shortest simple paths, often of potentially millions of paths, and then observed the number of paths that survive the edge deletion process through simulation.  Of particular interest, we found that the edge deletion process on the snapshot of the AS Graph looked remarkably similar to the simulations on realizations in the Chung-Lu random graph model with the appropriate expected degree sequence, further supporting the notion that Chung-Lu random graphs can emulate many of the properties observed in real world networks. 

Ultimately to find an efficient solution to the almost shortest path problem on real world networks, we need to consider the performance of the algorithm on plausible networks.  In this work, we not only provided an efficient solution to the almost shortest paths problems in terms of an important parameter of the problem for real world networks, the actual lengths of the paths, but also provided rigorous results relevant to the efficiency of using an almost shortest (nonbacktracking) path algorithm to find the almost shortest {\em simple} paths for realizations of the Chung-Lu random graph model, a model that captures many of the qualities empirically observed in real world networks.  


\section{Acknowledgements}
The authors would like to express their gratitude to Rachel Kartch and Rhiannon Weaver (CMU) for the many helpful conversations pertaining to the applications of the $k$ shortest path algorithms in context to internet routing in the Autonomous System graph.  Furthermore, the authors would also like to thank Philip Garrison (CMU) for his input in terms of constructing efficient $k$ shortest path algorithms for real world networks. \newline\newline
\noindent Copyright 2016 Carnegie Mellon University
\newline\newline
This material is based upon work funded and supported by Department of Homeland Security under Contract No. FA8721-05-C-0003 with Carnegie Mellon University for the operation of the Software Engineering Institute, a federally funded research and development center sponsored by the United States Department of Defense.
\newline\newline
NO WARRANTY. THIS CARNEGIE MELLON UNIVERSITY AND SOFTWARE ENGINEERING INSTITUTE MATERIAL IS FURNISHED ON AN “AS-IS” BASIS. CARNEGIE MELLON UNIVERSITY MAKES NO WARRANTIES OF ANY KIND, EITHER EXPRESSED OR IMPLIED, AS TO ANY MATTER INCLUDING, BUT NOT LIMITED TO, WARRANTY OF FITNESS FOR PURPOSE OR MERCHANTABILITY, EXCLUSIVITY, OR RESULTS OBTAINED FROM USE OF THE MATERIAL. CARNEGIE MELLON UNIVERSITY DOES NOT MAKE ANY WARRANTY OF ANY KIND WITH RESPECT TO FREEDOM FROM PATENT, TRADEMARK, OR COPYRIGHT INFRINGEMENT.
\newline\newline
[Distribution Statement A] This material has been approved for public release and unlimited distribution. Please see Copyright notice for non-US Government use and distribution.
\newline\newline
CERT \textregistered  is a registered mark of Carnegie Mellon University. DM-0003796

\bibliography{shortestpaths}

\end{document}